\documentclass[12pt]{amsart}

\usepackage{amsmath,amssymb}
\usepackage[dvipdfmx]{graphicx}
\usepackage[all]{xy}
\usepackage{xcolor}

\makeatletter
\@addtoreset{equation}{section}
\makeatother

\usepackage{enumerate}

\usepackage{color}

\def\ZZ{{\mathbb Z}}

\def\cC{{\mathcal C}}

\def\cT{{\mathcal T}}

\def\fc{{\mathfrak c}}

\def\fr{{\mathfrak r}}

\def\tE{{\widetilde E}}

\def\tG{{\widetilde G}}

\def\tV{{\widetilde V}}

\def\te{{\widetilde e}}

\def\tv{{\widetilde v}}

\def\BG{{\overline G}}

\def\hf{{\widehat f}}

\def\hati{{\widehat i}}

\def\hp{{\widehat p}}
\def\hq{{\widehat q}}

\def\book#1{\rm{#1}, }
\def\paper#1{\textit{#1}, }
\def\jour#1{\rm{#1}, }
\def\yr#1{({\rm{#1}) }}
\def\vol#1{\textbf{#1}}
\def\pages#1{\rm{#1}}

\def\publ#1{\rm{#1}, }
\def\by#1{{\rm{#1}, }}

\newtheorem{theorem}{Theorem}[section]
\newtheorem{definition}[theorem]{Definition}

\newtheorem{proposition}[theorem]{Proposition}
\newtheorem{corollary}[theorem]{Corollary}
\newtheorem{remark}[theorem]{Remark}
\newtheorem{lemma}[theorem]{Lemma}
\newtheorem{law}[theorem]{Law}

\def\book#1{\rm{#1}, }
\def\paper#1{\textit{#1}, }
\def\jour#1{\rm{#1}, }
\def\yr#1{({\rm{#1}) }}
\def\vol#1{\textbf{#1}}
\def\pages#1{\rm{#1}}

\def\publ#1{\rm{#1}, }
\def\by#1{{\rm{#1}, }}

\begin{document}

\title{Conway's law, revised from a mathematical viewpoint}

\author{Shigeki Matsutani, Shousuke Ohmori, Kenji Hiranabe, and Eiichi Hanyuda}

 \date{\today} 

\begin{abstract}
In this article, we revise Conway's Law from a mathematical point of view.
By introducing a task graph, we first rigorously state Conway's Law based on the homomorphisms in graph theory for the software system and the organizations that created it.
Though Conway did not mention it, the task graph shows the geometric structure of tasks, which plays a crucial role.
Furthermore, due to recent requirements for high-level treatment of communication (due to security, knowledge hiding, etc.) in organizations and hierarchical treatment of organizations, we have reformulated these statements in terms of weakened homomorphisms, and the continuous maps in graph topology.
In order to use graph topology and the continuous map in Conway's law, we have prepared them as mathematical tools, and then we show the natural expression of Conway's correspondences with hierarchical structures.
\end{abstract}

\maketitle

\section{Introduction}

Conway's law was given by a computer programmer Melvin E. Conway.
He stated, in 1967 \cite{Conway},
{\lq\lq}Organizations which design systems (in the broad sense used here) are constrained to produce designs which are copies of the communication structures of these organizations.{\rq\rq}

Conway represented the structure of the software system and the structure of the organization by geometric networks (graphs) and found the above law by the characteristics of the geometric structure of the relationship between the two graphs, though he called them {\lq}linear graph{\rq}.
It is that a complicated product would look like the organizational structure for which it was designed or engineered.
This law indicates that in order for a product to perform its desired function, it is necessary to ensure compatibility among its component parts, i.e., the relations among the software modules and the communication among the designers.

Based on these observations and laws, recent years have been spent studying how to structure an organization to produce the desired product.
The question has shifted to "What is the best way to produce the desired software systems?"
This is sometimes referred to as Conway's inverse problem.
Thus this problem has been studied in various fields \cite{HG1999, KCD2012, St2023}.
Skelton and Pais proposed a concept of team topologies \cite{SP2019}, which is of great interest.
They provided the many concrete strategies for the Conway's inverse problem, with respect to the size of teams, the individual cognitive load and so on.

Two of the authors of this article, mathematicians, heard about Conway's Law from, the other authors, agile researchers along with the book \cite{SP2019}.
As Conway used graphs in graph theory to represent the structures of the software systems and the organizations, it is natural to call such geometrical structures based on graph theory "topology."
However, topology, in general topology, is originally defined more rigorously, and from a rigorous mathematical point of view, some misuses of the concept seem inappropriate.

Further, Conway himself showed the relation between structures of the software systems and the organizations by a homomorphism in graph theory, which preserves a kind of graph (geometrical) structure, (see Definition \ref{def:homomorphism}.)
As the map in mathematics has the arrow (in general an inverse arrow is not a map), he considered a homomorphism from the graph of the system to the graph of organization. 
Prefacing with {\lq\lq}In the not unusual case where each subsystem had its own separate design group{\rq\rq}, he continued, {\lq\lq}we find that the structures (i.e., the linear graphs) of the design group and the system are identical{\rq\rq} \cite{Conway}.
Thus he did not consider at all the isomorphic graphs corresponding to the copies of the graph.
Indeed, he wrote
{\lq\lq}Notice that this rule is not necessarily one-to-one; that is, the two subsystems might have been designed by a single design group{\rq\rq} \cite{Conway}.
He introduced homomorphism a structure-preserving relationship after he wrote
{\lq\lq}In the case where some group designed more than one subsystem we find that the structure of the design organization is a collapsed version of the structure of the system, with the subsystems having the same design group collapsing into one node representing that group{\rq\rq} \cite{Conway}.

Many people who write softwares have experienced that it is possible to write software, no matter how complicated the structure, when developing a small software with a small number of people (or in the extreme, one person).
It is a natural question, whether it contradicts Conway's Law or not because the fact that if we view the graph structure described above, a complicated structure arises from a single node seems to contradict the first statement by Conway's, which is not a copy.
Conway's law is not applicable to such a case if we consider it strictly in the framework of homomorphism, because such an arrow is forbidden in homomorphism.
However, it can be expressed by a kind of {\lq}homomorphism{\rq} from the graph of the software structure to the edgeless graph, after slightly weakening the homomorphism, w-homomorphism (Definition \ref{df:w-Hom}), though it is not a strict homomorphism; w-homomorphism is related to the contraction of graphs \cite{Diestel}.
It depends on whether we allow internal communication in a team or a person; it corresponds to whether we consider a loop in the graph or the w-homomorphism.
It is unclear what Conway thought about these things, but they are certainly necessary concepts in discussing these software productions.

As Conway described the importance of hierarchical treatment in the section {\lq}the designed system{\rq}, homomorphism does not seem to be a good tool for expressing such states.
However the w-homomorphism works well as mentioned in (\ref{eq:Seq_GS_GO}) and Figure \ref{fg:HSseq_HOseq}.
By using the w-homomorphism, we can mathematically handle the hierarchical structure that Conway argued in the section.

However,  as Conway also noted, {\lq\lq}Even in a moderately small organization, it becomes necessary to restrict communication so that people can get some {\lq}{\lq}work{\rq\rq} done.{\rq\rq} \cite{Conway}.
His claim implies again that a structure that copies the structure of the system is important as an organizational structure.

\bigskip

There might exist a mathematical confusion.
Therefore, we consider that Conway's statement should be interpreted precisely from a mathematical viewpoint.

A more detailed analysis has been done \cite{SP2019}, but we believe that tools are needed to understand the modern requirements for sophisticated communications through security and other means, as well as more complex organizational charts and individual personnel levels.
The graph expression has hierarchical handling.
It requires a tool that is consistent with hierarchical handling.

The expressive power of topology in general topology is even more flexible, broad and applicable than the contrast between the two graphs.
Before we consider the detailed structure, we should rigorously express the geometric structures based on mathematics. After obtaining the mathematical foundation, we should proceed to discuss the detailed analysis from mathematical viewpoints.
This article is devoted to the mathematical foundation of the geometric structures associated with Conway's law.
Sometimes it is important to treat mathematics rigorously when making use of mathematics.
We attempt, in the spirit of Weil\footnote{It goes without saying that there are many examples in economy of the use of mathematics to solve problems in the social sciences \cite{Schofield}.
Besides them, one of the most famous applications of pure mathematics to the social sciences is Weil's analysis of the laws of marriage, which had a strong effect on structuralism \cite{Weil}.
This is an example where the structure of the Abelian group was clarified by strictly following pure mathematics rather than by complicated interpretations.
We believe that this successful example should be used as a model when applying mathematics to social structure and social sciences.
}, to understand this Conway's law according to strict topology rigorously following the general topology \cite{Kelley}, and give a clear mathematical interpretation, including the problems mentioned above.
We emphasize that rigorous topology in general topology has the potential to represent many mathematical structures well even in complicated materials geometry \cite{Ohmori}.

\bigskip

As Introduction, we show how to express Conway's law mathematically.
We consider it by three steps.
As a first step, we express rigorously in terms of homomorphism by introducing a task graph as intermediate state to settle the confusion that the restricting requirement to avoid unnecessary communication or cognitive load, and multiplicity of responsibility of the subsystems or modules of the software.
In the first step, we cannot treat the hierarchical structure because the homomorphism forbids it.

As a second step, we weaken the homomorphism to allow to consider the hierarchical structure by introducing w-homomorphism in Definition \ref{df:w-Hom}.
By using it, we can treat (\ref{eq:Seq_GS_GO}) and Figure \ref{fg:HSseq_HOseq}, which Conway described the importance and the meaning of by a single section.

However though it is sometimes important that we consider how the coarse structure is realized in the finer structure of the organization, the w-homomorphism is not a good tool to express it.
Further such hierarchical sequences (\ref{eq:Seq_GS_GO}) and Figure \ref{fg:HSseq_HOseq} should be described with their topology from a mathematical point of view.

Thus as a third step, we describe Conway's arguments by introducing the continuous maps and graph topology.
Since homomorphism and w-homomorphism are automatically continuous maps, we can naturally express Conway's law.

\bigskip

As Introduction, we partially show the first step.
Let the graphs of the artifact software system $S$ and the organization $O$ that produces $S$ be $G_S$ and $G_O$, respectively.

\subsection{Conway's law: mathematical revised as the first step}
We will mention it more precisely later.
\begin{enumerate}

\item Conway stated that there exists a homomorphism $f: G_S \to G_O$ by observations.

\item Suppose that for a given $G_O$ and an ideal software system $\tG_S$, which we require, we cannot find a homomorphism $f: \tG_S \to G_O$. 
Then the organization cannot produce the ideal system.

Then instead of $\tG_S$, we may find an appropriate system $G_S$ such that there exists an injective homomorphism $\iota : G_S \to \tG_S$ and a homomorphism $f : G_S \to G_O$. 
This means that the structure of the created software, $G_S$, is reflected by the structure of the organization, $G_O$, although this is not recommended.

In this article we called the subgraph $G_S$ of $\tG_S$ a truncated subgraph of $\tG_S$, which is a topological minor of $\tG_S$ \cite{Diestel}.

\item Suppose there is a homomorphism $f: G_S \to G_O$ which, for simplicity, is surjective as illustrated in Figure \ref{fg:GS_GT}.
We prepare a graph $G_T=(V_T, E_T)$ of the structure of the task of the software system such that there exists a surjective homomorphism $p: G_S \to G_T$ and injective homomorphism $i: G_T \to G_O$ satisfying $f = i \circ p$.

We call $p$ a fiber structure, which means that for a vertex $v_O$ of $G_O$, a team or a person, $v_O$, is responsible for the submodules $p^{-1}(\{i^{-1}(v_O)\})=f^{-1}(\{v_0\})$.

\end{enumerate}

Conway claimed that to avoid unnecessary communication, $i$ must be isomorphic and $G_T$ should not be complicated.
He also implicitly assumed that this $G_T$ must reflect the geometric structure of the software $S$ in order to have mathematical meaning as a modeling, which is expressed by the homomorphism.

As the second step, the above can be weakened by replacing the homomorphism with a weak one.

\bigskip

As the third step, we introduce the topology of graphs and consider the properties of the continuous map between them to express Conway's law more naturally from a mathematical point of view.
As Skelton and Pais developed it as {\lq}team topologies{\rq} \cite{SP2019}, it means that we clarify that the relation of Conway's law and topology.

Using them, we also propose that Conway's law should be described in accordance with hierarchical structures, after assigning the topology $\cT_S$ or $\cT_T$ and $\cT_O$ in general topology to the artifact software $S$ or task $T$ and the organization $O$ producing $S$.
Then there is a {\lq}morphism{\rq} from $O$ to $S$.

\begin{enumerate}

\item 
Since for a set $X$ a topology $\cT_X$ of $X$ is an assignment, we can assign a different topology $\cT_X'$ to $X$. 
Assigning a topology means choosing the {\lq}geometric structure{\rq} of $X$.
Thus we say that the structure in the studies of Conway's law should be understood as the assignment of topology by the graph structure.

\item
For a graph $G$, we show that the topology of $G$ in general topology is naturally determined by considering the collection of connected subgraphs of the graph $G$, since the collection of connected subgraphs of the graph $G$ is a partial order set, and the partial order set has the order topology, which recovers the order \cite{AN2023}.

\item 
For sets $X$ and $Y$ and a map $f : X\to Y$, whether $f$ is continuous or not is strictly defined in topology theory, even for finite sets $X$ and $Y$;
the continuous map $f$ is defined only for sets with topologies $\cT_X$ and $\cT_Y$ such that the inverse image $f^{-1}$ of any open set in $Y$ belongs to $\cT_X$.
We consider a map $f$ from the collection $\cC_{G_S}$ of connected subgraphs with the topology $\cT_{G_S}$ of the graph $G_S$ of the software to $\cC_{G_O}$ of that of the organization $O$.
If $f$ is a continuous injection, we can define the morphism $f^\#: G_O \to G_S$.
Instead of maps between graphs, we consider a map between the collection of subgraphs of graphs, we can flexibly deal with geometric structures of graphs.

\item For example, if $O$ has enough members with complete graphs and topology $\cT_O$, the map is continuous even for any $S$ with complicated topology.

\item For the topological spaces, we define the mathematical structures on them, known as sheaves in sheaf theory, which play a crucial role in modern mathematics \cite{Gogunen, Sendroiu}.

\end{enumerate}

Using these notes, we investigate the geometric structures of software and organizations based on the topology of the general topology.
Since the homomorphisms and the w-homomorphisms of the graphs are continuous maps, we can rephrase Conway's law above in terms of continuous maps.
By rephrasing it, we express the complicated structure of the organizations with hierarchical structure as the third step.

The contents of this article are as follows.
Section 2 is for mathematical preparation of the topological expression of graphs. In Section 3, we show the topological expression of graphs.
Section 4 is the mathematical investigation of the geometric structures associated with Conway's law.
We give a conclusion of this article in section 5.

\section{Mathematical preliminary}

\subsection{Review of topology}

We recall the definition of topology in the general topology.

\begin{definition}
For a set $X$ and its power set $\wp(X)$, i.e., set of any subset of $X$, a subset $\cT$ of $\wp(X)$ is topology of $X$ if the following hold.

\begin{enumerate}

\item $X$ and $\emptyset$ belong to $\cT$,

\item the union of the elements of any subcollection of $\cT$ belongs to $\cT$, and

\item the intersection of the elements of any finite subcollection of $\cT$ belongs to $\cT$.

\end{enumerate}
We refer to the element of $\cT$ as an open set.
In other words, $\cT$ is the set of open sets.
Further we call $(X, \cT)$ \emph{topological space}.
\end{definition}

For a topological space  $(X, \cT)$, we say that $U \in \wp(X)$ is a closed set if $U^\fc := X \setminus U$ belongs to $\cT$.
Let $\cT^{(c)}$ be the set of all closed sets.

\begin{remark}
{\rm{
A topological space $(X, \cT)$ of a finite set $X$ can be regarded as the set of closed sets.
}}
\end{remark}

We also recall the neighborhood of a point $q\in X$.

\begin{definition}
For a topological space $(X,\cT_X)$, take a point $q$ of $X$.

\begin{enumerate}
\item An \emph{open neighborhood} $U$ of $q$ is an open set and contains $q$,
i.e., $q\in U\in\cT_X$.

\item A subset $V \subset X$ is a \emph{neighborhood} of $q$ if $V$ contains an open neighborhood $U$ of $q$, $q \in U\subset V$ for $q\in U\in\cT_X$.

\item The \emph{closure} of $U$ denoted by $\overline{U}$ is defined by
$
\overline{U}:= \displaystyle{\bigcap_{V\in\cT_X^{(c)}: U\subset V} V}
$.

\item 
Suppose $X$ is a finite set.
Let $U \in \cT_X$ be a minimal open neighborhood of $q$, i.e., 
$\displaystyle{U:=\bigcap_{V\in\cT_X: q\in V} V}$, which is denoted by $\underline{q}$.  

\end{enumerate}

\end{definition}

\bigskip

\subsection{Review of poset}

Let us consider a partial ordered set (poset) $(P, \le)$; 
$P$ is a set satisfying the following:
there is a map $\fr_{\le}$ from $P \times P$ to $\ZZ_2:=\{$false, true$\}$ for $\le$ satisfying

\begin{enumerate}
\item for any $x\in P$, $x \le x$ (\emph{reflexive relation}),

\item  if $x$ and $y \in P$ satisfy $ x \le y$ and $y \le x$, $x = y$ 
(\emph{anti-symmetric relation}) and

\item
 if $x,y,z\in P$ satisfy $ x \le y$ and $y \le z$,  $x \le z$ (\emph{transitive relation}).
\end{enumerate}

If the cardinality of a poset $P$ is finite, we call it a \emph{finite poset}.

\begin{lemma}\label{lm:poset_top}
A finite poset $(P, \le)$ has a natural topology $\cT_P$ induced from $\leq$, i.e., the set $P$ with $\cT_P$ recovers the partial order $\leq$ again.
The topological space $(P, \cT_P)$ is referred to as \emph{the poset topology} of $(P, \le)$.
\end{lemma}

\begin{proof}
For $x \in P$, we let $U_x:=\{y \in P\ |\ y \le x\}$.
We note that for $y \in U_x$, $U_y \subset U_x$.
Now we introduce $\cT_0:=\{ U_x \ | \ x \in P\}$.
We define $\cT_P$ generated by $\cT_0$, or
$$
\cT_P:=\{\cup_{\lambda} U_\lambda \ |\ U_\lambda \in \cT_0\} \cup \emptyset.
$$
Then it is obvious that $\cT_P$ satisfies the definition of topology.

For any $q$ and $p$ in $P$, we will define the map $\fr_{\le}(p, q)$ to $\ZZ_2$ as follows.
Let their minimal open neighborhoods $\underline{p}$ and $\underline{q}$ of $p$ and $q$.
Then if $\underline{p} \subset \underline{q}$, let $p \leq q$ and if $\underline{q} \subset \underline{p}$, $q \leq p$.
If there is no such relation between $\underline{p}$ and $\underline{q}$, we don't define the relation between $p$ and $q$.
Then we recover the partial order for $P$.
From the construction, it is consistent with the original one.
\end{proof}

Note that in $(P, \cT_P)$, we have $\underline{p}=U_p$ for every $p\in P$ by using the notation $U_x$ in the above proof.

\begin{lemma}\label{lm:4cases}
We consider a topological space $(P,\cT_P)$ induced from $\leq$ in a finite poset $(P, \le)$.
For any $p$ and $q$ in $P$, there are four disjoint cases:
\begin{enumerate}
\item $p,q \in \underline{q}\cap \underline{p}$ which corresponds to $p = q$,

\item $p \in \underline{q}\cap \underline{p}$ and $q \not\in \underline{q}\cap \underline{p}$ which corresponds to $p \leq q$ and $p\neq q$

\item $q \in \underline{q}\cap \underline{p}$ and $p \not\in \underline{q}\cap \underline{p}$ which corresponds to $q \leq p$ and $p\neq q$, and 

\item $p, q \not\in \underline{q}\cap \underline{p}$ which corresponds to $p \not\leq q$, and $q \not\leq p$.

\end{enumerate}
\end{lemma}

\begin{proof}
\begin{enumerate}

\item If $p,q\in \underline{q}\cap \underline{p}$, then we have $\underline{p}=\underline{q}$, which means $p=q$. 
		For $p=q$, the relation $p,q\in \underline{q}\cap \underline{p}$ is clear.
		
\item $p\in \underline{q}\cap \underline{p}$ shows $p \leq q$. If $p=q$, $q\in \underline{q}\cap \underline{p}$.
		Conversely, if $p\leq q$ and $p\not =q$, then $p \in \underline{p}\subset \underline{q}$. Thus, $p\in \underline{q}\cap \underline{p}$.
		If $q\in \underline{q}\cap \underline{p}$, we obtain $q\leq p$ and hence, $p=q$. 

\item It is clear from the proof of (2).

\item Suppose $p \leq q$ holds. 
		Then, $p\in \underline{q}$ that contradicts to $p,q \not \in \underline{q}\cap \underline{p}$.
		Also, we leads to a contradiction if $q \leq p$ holds.
		Conversely, if $p \in \underline{q}\cap \underline{p}$, then $p\in \underline{q}$. Hence, $p\leq q$, which contradicts $p\not \leq q$ and $q\not \leq p$.
		If $q \in \underline{q}\cap \underline{p}$, we have $q\leq p$ that is a contradiction.
\end{enumerate}
\end{proof}

We have a key lemma:

\begin{lemma}\label{lm:poset_cont}
Let $(P, \le)$ and $(P', \le ')$ be posets.
We consider the poset topological spaces $(P, \cT_P)$ and $(P', \cT_{P'})$ induced from $\leq$ and a map $f: P \to P'$.
If and only if $f$ is continuous, for any $p$ and $q$ in $P$ such that $p \leq q$, $f(p) \leq f(q)$.
\end{lemma}

\begin{proof}
Suppose that $f:P\to P'$ is continuous.
Fix $q\in P$.
From $f(q)\in \underline{f(q)}$, we have $q\in f^{-1}(\underline{f(q)})\in \cT_P$.
The minimality of $\underline{q}$ provides $\underline{q}\subset f^{-1}(\underline{f(q)})$, which implies $f(\underline{q}) \subset \underline{f(q)}$.
Therefore, if $p \leq q$, then $p\in \underline{q}$ and by $f(\underline{q}) \subset \underline{f(q)}$ we obtain $f(p)\leq f(q)$.

Since in general, for a map $f: X\to Y$, $f^{-1}(\cup_{\lambda \in \Lambda} A_{\lambda}) = \cup_{\lambda \in \Lambda} f^{-1}(A_{\lambda})$, we show that for every $q'\in P'$, $f^{-1}(\underline{q'})$ is an open set in $P$.
If $f^{-1}(\underline{q'})$ is empty, it is obvious.
Thus we assume that $f^{-1}(\underline{q'})\neq \emptyset$.
Take an element $p \in f^{-1}(\underline{q'})$. From the definition, $f(p)\leq q'$. For any $q\in \underline{p}$ (i.e., $q \leq p$, $f(q)\leq f(p)$, and thus 
$f(q)\leq f(p)\leq q'$. It means that $q \in f^{-1}(\underline{q'})$, and $\underline{p}\subset f^{-1}(\underline{q'})$.
Hence $f^{-1}(\underline{q'})$ is an open set in $P$.
\end{proof}

\begin{corollary}
Let $(P, \le)$ and $(P', \le ')$ be posets with the poset topologies $(P, \cT_P)$ and $(P', \cT_{P'})$ and a continuous map $f: P \to P'$.
Then the following hold.

\begin{enumerate}

\item for every $q \in P$, $f(\underline{q}) \subset \underline{f(q)}$, and

\item for every $p', q' \in f(P)$, such that $p' \not\leq' q'$, we have $p \not\leq q$ for every $p \in f^{-1}(\{p'\})$ and $q \in f^{-1}(\{q'\})$.

\end{enumerate}
\end{corollary}

\begin{proof}
We consider 1. 
The continuity of $f$ means that $f^{-1}(\underline{f(q)})$ is an open subset of $P$. Obviously $q \in f^{-1}(\underline{f(q)})$.
$\underline{q}$ is the smallest open neighborhood of $q$.
Thus $\underline{q}\subset f^{-1}(\underline{f(q)})$.
Hence, we obtain the relation 1.

We consider 2.
Under the assumptions, we assume that there exist $p \in f^{-1}(\{p'\})$ and $q \in f^{-1}(\{q'\})$ such that $p {\leq} q$, or $p \in \underline{q}$.
Hence $f(p) \in f(\underline{q})$.
From 1, $p'=f(p) \in f(\underline{q})\subset \underline{f(q)}=\underline{q'}$.
It contradicts the assumptions.
\end{proof}

\subsubsection{Examples of the continuous map}

\begin{enumerate}

\item For every the poset $(P, \leq)$, and a unit set $\{x\}$,
 every map $f : (P, \leq)\to \{x\}$ is continuous.

\item For every the poset $(P, \leq)$, and a unit set $\{x\}$,
 every map $f : \{x\}\to (P, \leq)$ is continuous.

\item Let $P=\{a, b_1, b_2, c_1, c_2, c_3. c_4\}$, with 
$c_1 \leq b_1$, $c_2\leq b_1$, $c_3 \leq b_2$, $c_4\leq b_2$, $b_1\leq a$, and $b_2\leq a$, and $P'=\{A, B, C\}$ with $C\leq B$, $B \leq A$.
Then $f: P \to P'$ ($f(a)=A$, $f(b_i)=B$, $f(c_i)=C$) is continuous.

\item Assume that $\{0_i\ |\ i = 1, 2, 3\}$  and $\{1_i\ |\ i = 1, 2, 3\}$ such that $0_i \leq 1_i$.
We consider $P=\{j_i \ |\ j=0, 1, i=1,2\}$ and $P'=\{0_3, 1_3\}$ and a map $f : P \to P'$, $(f(j_1)=f(j_2)=j_3$ for $j=0,1$. Then $f$ is continuous.

\end{enumerate}

\begin{remark}
{\rm{
Let $(P, \le)$ and $(P', \le ')$ be posets with the poset topologies $(P, \cT_P)$ and $(P', \cT_{P'})$ and a continuous injection $f: P \to P'$.
The statement that for any $p, q \in P$ such that $f(p) \leq f(q)$, $p \leq q$ does not hold in general.

We have the counter examples.

\begin{enumerate}

\item Let $P=\{a, b, c_1, c_2, c_3\}$, with 
$c_1 \leq a$, $c_2\leq b$, $c_3 \leq b$, $b\leq a$,
and 
$P'=\{a', b', c'_1, c'_2, c'_3\}$, with 
$c_1' \leq b'$, $c_2'\leq b'$, $c_3' \leq b'$, $b'\leq a'$.
The map $f: P \to P'$ ($f(x)=x'$) is continuous because
$f^{-1}(\underline{b'})=\underline{b}\cup \underline{c_1}\in \cT_P$.
We have $f(c_1) \leq f(b)$ but $c_1 \not \leq b$.

\item 
Let $P=\{a, b, c_1, c_2, c_3\}$, with 
 $c_2\leq b$, $c_3 \leq b$, $b\leq a$,
and 
$P=\{a', b', c'_1, c'_2, c'_3\}$, with 
$c_1' \leq b'$, $c_2'\leq b'$, $c_3' \leq b'$, $b'\leq a'$.
The map $f: P \to P'$ ($f(x)=x'$) is also continuous because
$f^{-1}(\underline{b'})=\underline{b}\cup \underline{c_1}\in \cT_P$.
We have $f(c_1) \leq f(b)$ but $c_1 \not \leq b$.

\end{enumerate}
}}
\end{remark}

\begin{lemma}\label{lm:localhome}
Let $(P, \le)$ and $(P', \le ')$ be posets with the poset topologies $(P, \cT_P)$ and $(P', \cT_{P'})$ and a continuous injection $f: P \to P'$.
We assume that for any $p, q \in P$ such that $f(p) \leq f(q)$, $p \leq q$,
Then the following hold.

\begin{enumerate}

\item for $q \in P$, $\underline{f(q)} \cap f(P)=f(\underline{q})$, 

\item the inverse map $f^{-1}:f(P)\to P$ is continuous, and 

\item for any $p, q \in P$ such that $p \not\leq q$, $f(p) \not\leq f(q)$.

\end{enumerate}
\end{lemma}

\begin{proof}
We consider the first statement.
The continuity shows that $\underline{f(q)} \cap f(P) \supset f(\underline{q})$.
We show that $\underline{f(q)} \cap f(P) \subset f(\underline{q})$.
Since $f$ is injection, for every $r' \in \underline{f(q)} \cap f(P)$, we find a unique $r$ such that $f(r)=r'$.
From the first statement, $f(r) \leq f(q)$ means $r \leq q$, and thus $r \in \underline{q}$. 
Thus $r' \in f(\underline{q})$.
It means that we have the equality.

We consider the second statement.
$f(P)$ is a sub-poset of P'. We also consider it as a topological space.
The basis of the topologies are bijective. Thus it is obvious.

We consider the third statement that for $p \not\leq q$, i.e., $p \not\in \underline{q}$, we have $f(p) \not \in f(\underline{q})$ and thus $f(p) \not\leq f(q)$.
\end{proof}

\begin{definition}
We refer to the $P$, $P'$ and $f: P \to P'$ satisfying the condition in Lemma \ref{lm:localhome} as an embedding map.

Further, we say that $P$ and $P'$ are homeomorphic if there is a continuous map $f: P \to P'$ is bijection and its inverse map $f^{-1}$ is also continuous; 
we refer to the map as a homeomorphic map.
\end{definition}

\begin{remark}
{\rm{
The embedding map in the topological space $(P, \cT_P)$ means the order preserving map between the posets as the above sense.

It is obvious that if an embedding map $f: P \to P'$ is bijection, $f$ is a homeomorphic map.
}}
\end{remark}

\section{Topological expression of graphs}

\subsection{Review of graph theory}

In this section, we give the topological expression of graphs after we recall the homomorphism and introduce the weakened one.

Let $G=(V, E)$ be a graph,  where $V$ is the set of vertices and $E$ is the set of edges of $G$.
In this paper we only consider simple graphs, i.e., graphs without loops or multiple edges.
Sometimes we denote the set of vertices of a graph $G$ by $V(G)$ and the set of edges by $E(G)$.
If $G'$ is a subgraph of $G$, we write $G' \subset G$.

We recall the graph homomorphism \cite{BR}, which is used in Conway's paper \cite{Conway}.
\begin{definition}\label{def:homomorphism}
Let $G=(V,E)$ and $G'=(V',E')$ be simple graphs, i.e., without loops and multiple edges. 
A homomorphism from $G$ to $G'$ is a map $f:V\to V'$ such that $(f(v), f(v')) \in E'$ whenever $(v, v')\in E$.
The map $f$ is an isomorphism if $f$ is bijective and $(v, v') \in E$ if and only if $(f(v), f(v')) \in E'$.
\end{definition}

Since we are dealing with simple graphs, we note that $(v,v)$ does not belong to $E$ for $G$ in the definition of homomorphism.

In this article, we also prepare the following.

\begin{definition}\label{def:homomorphism2}
Let $G=(V,E)$ and $G'=(V',E')$ be simple graphs.
\begin{enumerate}
\item If $i: V \to V'$ is injective, and $i$ is homomorphism from $G$ to $G'$, then we say that the homomorphism $i$ is injective, and we call $G$ a truncated graph of $G'$.

\item If $p: V \to V'$ is surjective, $p$ is homomorphism from $G$ to $G'$, and $\{(p(v), p(v'))\ $ $| (v, v') \in E\}=E'$, then the homomorphism $p$ is surjective, and we call $G$ a fibered graph of $G'$.
\end{enumerate}
\end{definition}

The truncated graph is sometimes called {\lq}topological minor{\rq} \cite{Diestel}.
In \cite{BR}, the homomorphism $f$ from $G$ to $G'$ such that $f(V(G))=V(G')$ is called the onto-homomorphism, which differs from the above definition of the surjection.

Further we introduce weak homomorphisms.
\begin{definition}\label{df:w-Hom}
Let $G=(V,E)$ and $G'=(V',E')$ be simple graphs.
\begin{enumerate}
\item A w-homomorphism from $G$ to $G'$ is a map $f:V\to V'$ such that $f(v)=f(v')$ or $(f(v), f(v')) \in E'$ whenever $(v, v')\in E$.

\item If $i: V \to V'$ is injective, and $i$ is w-homomorphism from $G$ to $G'$, then then we say that the w-homomorphism $i$ is w-injective, and we call $G$ a  w-truncated graph of $G'$.

\item If $p: V \to V'$ is surjective, $p$ is w-homomorphism from $G$ to $G'$, and $\{(p(v), p(v'))\ $ $| (v, v') \in E, p(v)\neq p(v')\}=E'$, then we say that the w-homomorphism $p$ is surjective and w-projective, and we call $G$ a w-fibered graph of $G'$.
\end{enumerate}
\end{definition}

In this paper we mainly consider the connected graphs, and so we don't need the condition that $p: V \to V'$ is surjective on the w-projection and the projection when we take a connected graph.

The following are obvious.
\begin{proposition}\label{pr:Hom_vs_wHom}
Let $G$, $G'$, $G''$ be simple graphs such that there exist homomorphisms (w-homomorphism) $f: G \to G'$ and $g : G' \to G''$.
Then $g\circ f: G\to G''$ is a homomorphism (w-homomorphism).

If $f: G \to G'$ is a homomorphism, $f$ is a w-homomorphism.

If $f: G \to G'$ is an injective w-homomorphism, $f$ is an injective homomorphism.
\end{proposition}

The surjective w-homomorphism is related to the contraction of graphs \cite{Diestel}.

\begin{lemma}
For any simple graph $G=(V, E)$, $V\neq \emptyset$, the edgeless graph  is a truncated and w-fibered graph of $G$.
\end{lemma}

\begin{proof}
There exist a homomorphisms $i : (\{v\}, \emptyset) \to G$ and a w-homomorphism $p :G \to (\{v\}, \emptyset)$.
\end{proof}

The homomorphism is characterized by the following proposition.
\begin{proposition}
Let $G=(V,E)$ and $G'=(V',E')$ be simple graphs such that there is a homomorphism $f: G \to G'$.
For every $v' \in V'$, there is no pair $(v_1, v_2)$ in $f^{-1}(\{v'\})$ such that $(v_1, v_2) \in E$.
\end{proposition}

\begin{proof}
If there is such an edge $(v_1, v_2)$, there is a loop $(v', v') \in E'$.
\end{proof}

We show two examples, a case that there is a homomorphism, and a case that there is a w-homomorphism but not homomorphism in Figure \ref{fg:Hom_vs_wHom}.

\begin{figure}
\begin{center}
\includegraphics[width=0.35\hsize]{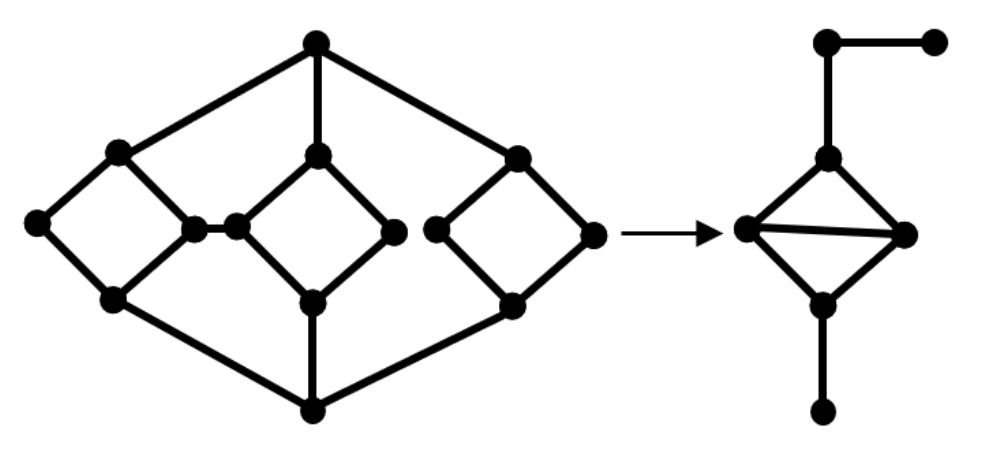}
\hskip 0.1\hsize
\includegraphics[width=0.35\hsize]{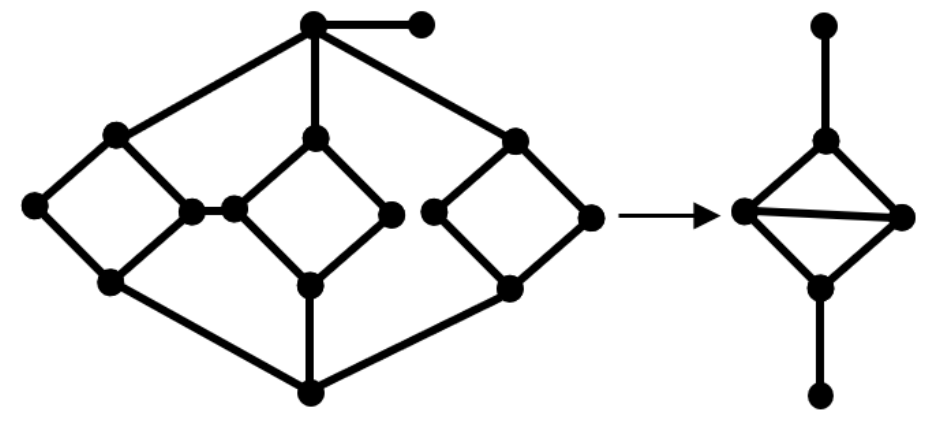}

(a) \hskip 0.5\hsize (b)

\end{center}

\caption{
(a) There is a homomorphism, but (b) there is a w-homomorphism but not homomorphism.
}\label{fg:Hom_vs_wHom}
\end{figure}

\begin{remark}\label{rmk:Hom_wHome}
{\rm{
We comment on the relation between homomorphism and w-homomorphism.

For a simple graph $G=(V, E)$, we consider a graph with loops $(v, v)$ for all $v\in V$.
Let $E^\circ:= E\cup\{(v, v)\ | v\in V\}$ and $G^\circ:=(V, E^\circ)$.
Then we naturally define the map $\varpi : G^\circ \to G$: 
$\varpi|_{(V, E)}=\mathrm{id}$ and $\varpi((v,v))\mapsto v$.

If we apply the statement in the definition of homomorphism to a correspondence $G \to G^{\prime\circ}$, i.e., that for $(v, v') \in E$, $(f(v), f(v')) \in E^{\prime \circ}$, and then we consider its projection by $\varpi$, we obtain the w-homomorphism.
In other words, when we go beyond simple graphs to non-simple graphs, the difference between homomorphism and w-homomorphism disappears.
}}
\end{remark}

\begin{remark}\label{rmk:Hom_wHome2}
Using the w-homomorphism, we may consider a w-projective sequence of graphs.
\begin{equation}
\xymatrix{ 
G_{s}=G_0 \ar[r]^{p_0} & \cdots  \ar[r]^{p_{i-1}}& G_i \ar[r]^{p_i} & G_{i+1}
 \ar[r]^{p_{i+1}}& {\cdots} \ar[r]^{p_t\ \ \ \  } & G_{t}=(\{v\},\emptyset),
 }
\end{equation}
where $p_i$ is a surjective w-homomorphism.

Then we can handle the morphism between the sequences by the commutative diagram.
\begin{equation}
\xymatrix{ 
G_{s}=G_0 \ar[r]^{p_0}\ar[d]^{f_0} & \cdots  \ar[r]^{p_{i-1}}& G_i  \ar[r]^{p_i}\ar[d]^{f_{i}} & G_{i+1} \ar[r]^{p_{i+1}}\ar[d]^{f_{i+1}} 
&{\cdots} \ar[r]^{p_t\ \ \ \  } &  G_{t}=(\{v\},\emptyset)\ar[d]^{f_{t}} \\
G'_{s}=G'_0 \ar[r]^{p'_0} & \cdots  \ar[r]^{p'_{i-1}}& G'_i  \ar[r]^{p'_i}& G'_{i+1} \ar[r]^{p'_{i+1}} &{\cdots}  \ar[r]^{p'_t\ \ \ \ } & G'_{t}=(\{v'\},\emptyset),
 }
\end{equation}
where $p_i$ and $p'_i$ are surjective w-homomorphisms and $f_i$ is a homomorphism or w-homomorphism.
When we consider such sequences, we cannot express them only by the homomorphisms in general.

\end{remark}

\subsection{Topology of graph}

To handle the relations between graphs more flexibly, we introduce the topology of graphs for finite graphs.

\begin{lemma}
For a finite connected graph $G$, we define the set $\cC_G$ as the collection of connected subgraphs of $G$. 
Then $\cC_G$ is a poset for $\subset$.
\end{lemma}

\begin{definition}
For a finite connected graph $G$ and its collection of the connected subgraphs $\cC_G$, which is a poset with respect to $\subset$,
let the poset topology on $\cC_G$ induced from $\subset$ be denoted by $\cT_G$
and call it the graph topology of $G$. 
The topological space is denoted by $(G, \cC_G, \cT_G)$.
\end{definition}

We note that the this definition slightly differs from \cite{AN2023}.

In this article, we handle only connected simple graphs, without loops and multiple edges and thus, we sometimes refer to them as graphs without stating the connected and simple conditions hereafter.
Examples of the collection of connected subgraphs with the partial order $\subset$ are illustrated in Figure \ref{fg:SubGrap1}.

\begin{figure}
\begin{center}
\includegraphics[width=0.15\hsize]{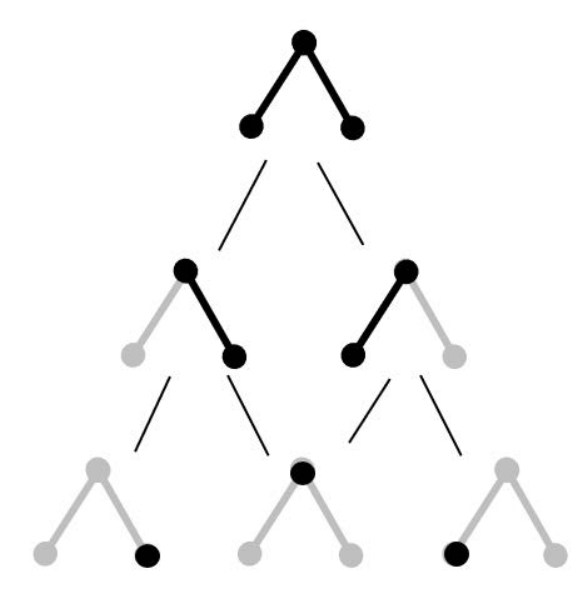}
\hskip 0.1\hsize
\includegraphics[width=0.40\hsize]{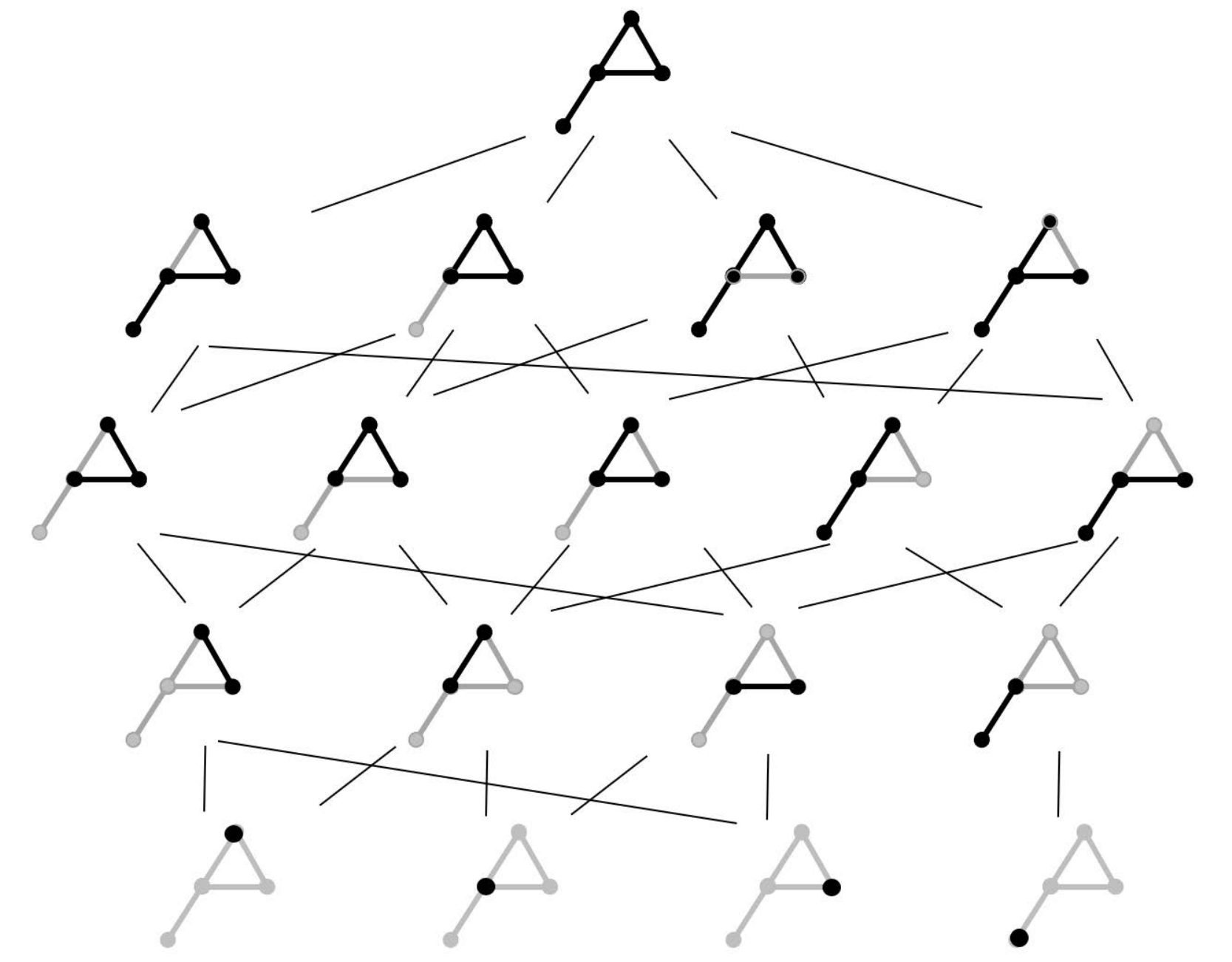}

(a) \hskip 0.3\hsize (b)

\includegraphics[width=0.15\hsize]{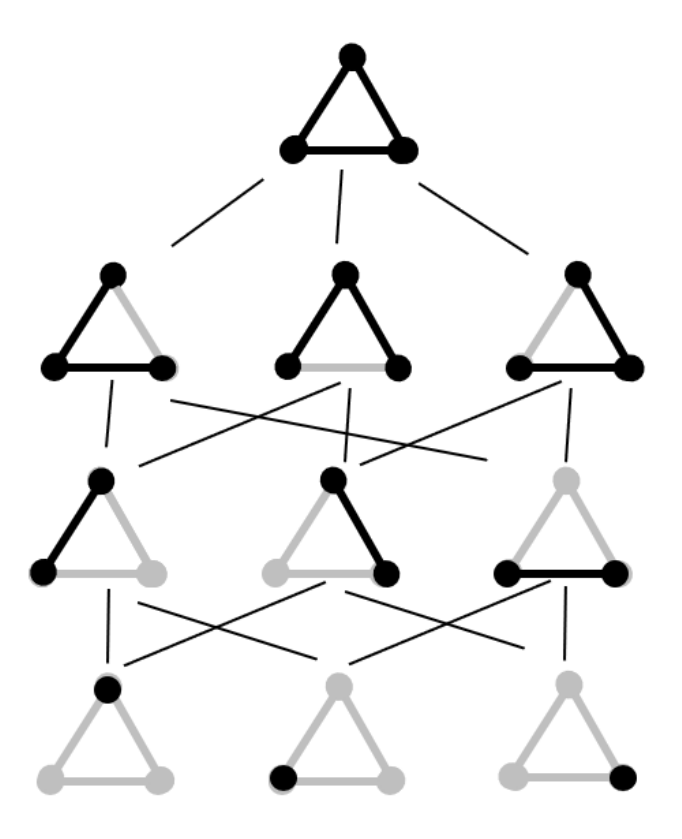}
\hskip 0.1\hsize
\includegraphics[width=0.25\hsize]{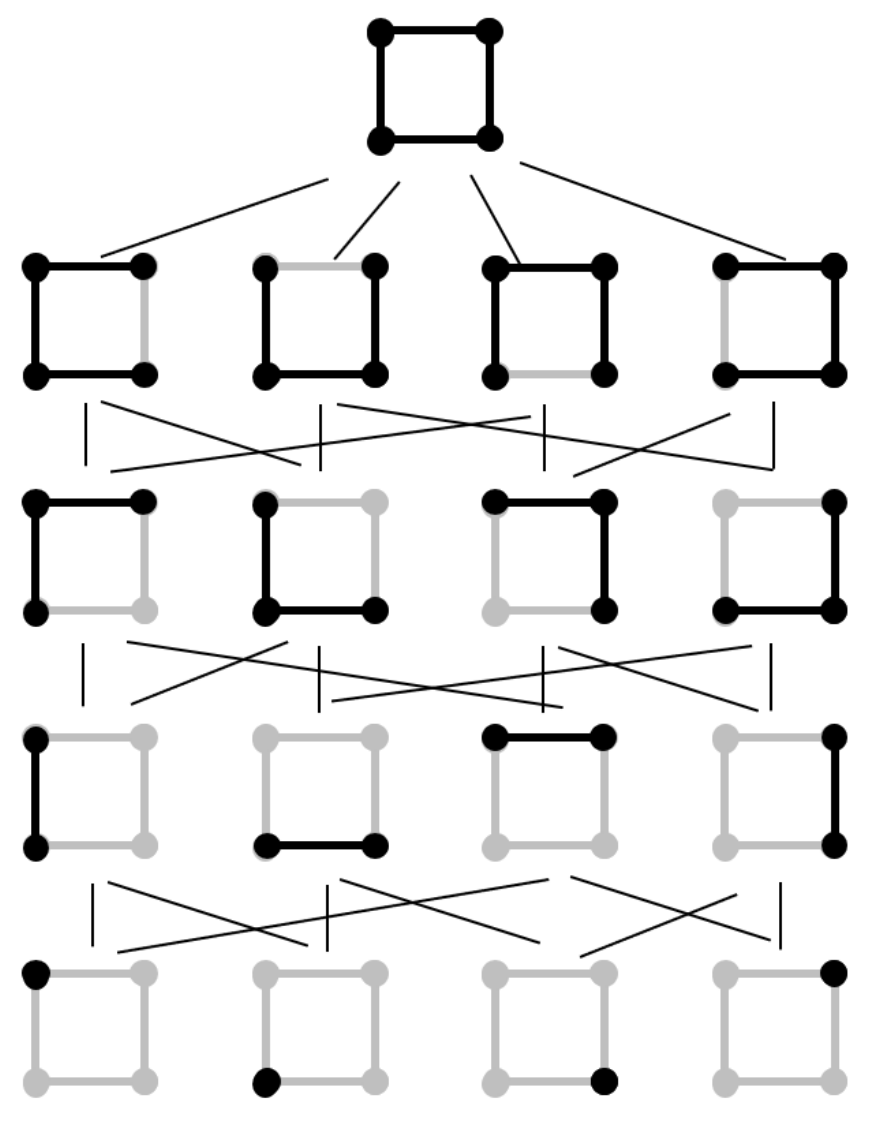}
\hskip 0.1\hsize
\includegraphics[width=0.30\hsize]{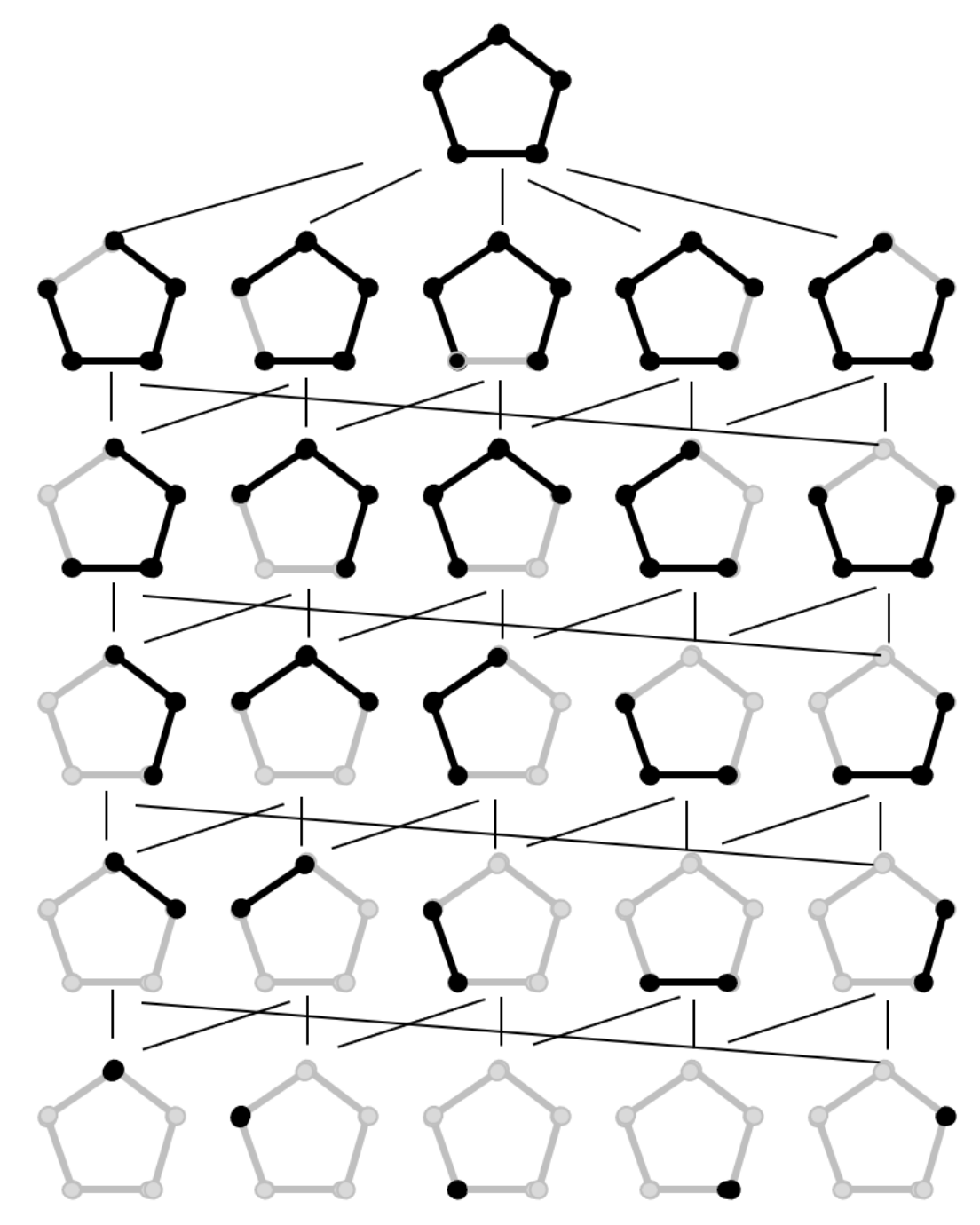}

(c) \hskip 0.3\hsize (d)
\hskip 0.3\hsize
 (e)

\end{center}

\caption{
Examples of the collection of connected subgraphs with the partial order $\subset$ are illustrated.
}\label{fg:SubGrap1}
\end{figure}

\begin{lemma}
For a given $\cC_G$, we recover $G$ as a pair of sets $(V,E)$.
\end{lemma}

\begin{proof}
In $\cC_G$, there exists $G$, all edgeless graphs $(\{v\},\emptyset)$ for every $v\in V$ and all single edge graph $(\{v_1, v_2\}, \{e\})$ for every $e=(v_1, v_2)\in E$. In other words, $\cC_G$ contains $V$ and $E$ as subgraphs of $G$.
\end{proof}

Lemma \ref{lm:poset_top} leads us to have the following lemma.

\begin{lemma}
For a given $\cT_G$, we recover $\cC_G$.
\end{lemma}

We introduce morphisms from $(G, \cC_G, \cT_G)$ to $(G', \cC_{G'}, \cT_{G'})$.

\begin{definition}
Let $(G, \cC_G, \cT_G)$ and $(G', \cC_{G'}, \cT_{G'})$ be the finite connected graphs with their graph topologies.
\begin{enumerate}

\item A map $f :\cC_G \to \cC_{G'}$ is continuous, if $f$ is continuous as the poset  $(\cC_G, \subset)$ and $(\cC_{G'}, \subset)$.

\item A map $f :\cC_G \to \cC_{G'}$ is a continuous injection, if $f$ is a continuous injection as the poset  $(\cC_G, \subset)$ and $(\cC_{G'}, \subset)$.

\item A map $f :\cC_G \to \cC_{G'}$ is embedding, if $f$ is embedding as the poset  $(\cC_G, \subset)$ and $(\cC_{G'}, \subset)$.

\item A bijection $f :\cC_G \to \cC_{G'}$ is homeomorphic, if $f$ is embedding.
If there is a homeomorphic map between $\cC_G$ and  $\cC_{G'}$, we say that $\cC_G$ and  $\cC_{G'}$ are homeomorphic.

\end{enumerate}
\end{definition}

We illustrate some examples of the continuous injection in Figure \ref{fg:SubGrap2} and a pair that there is no embedding in Figure \ref{fg:SubGrap3}.

Following three lemmas are obvious.

\begin{lemma}\label{lm:singleVer}
For an edgeless graph $G_v:=(\{v\},\emptyset)$ and any simple graph $G$, a map $f$ from $\cC_{G_v}$ to the collection connected subgraphs $\cC_G$ of $G$ is continuous.
\end{lemma}

\begin{lemma}\label{lm:singleVer2}
For an edgeless graph $G_v:=(\{v\},\emptyset)$  and any simple graph $G$, a map $f$ from the collection connected subgraphs $\cC_G$ of $G$ to $\cC_{G_v}$ is continuous.
\end{lemma}

\begin{lemma}\label{lm:subGcont}
Let $G$ be a subgraph of $G'$, i.e., $G \subset G'$
Then we can find a continuous injection $f$ from $\cC_{G}$ to $\cC_{G'}$.
\end{lemma}

\begin{proof}
We choose $f(G)$ such that $f(G)$ is the subgraph $G$ in $G'$.
\end{proof}

\begin{lemma}\label{lm:subGcontII}
Let $G$ and $G'$ be graphs such that there is a continuous injection $f$ from $\cC_{G}$ to $\cC_{G'}$.
Then in general $G$ is not isomorphic to a subgraph in $G'$, or, 
there is not a homomorphism from $G$ to $G'$.
\end{lemma}

\begin{proof}
Figure \ref{fg:SubGrap2} is an example of the statement but we show a simpler one.

Let $G''$ be isomorphic to $G$. 
Then we add a vertex in an edge $e$ of the subgraph $H$ in $G$, i.e., for $e=(v_i, v_j)$, we consider $e'=(v_i, v)$ and $e''=(v, v_j)$ by inserting $v$ into $e$, and refer the generated graph of $G''$ as $G'$, which is not isomorphic to $G$. Further, there is no homomorphism from $G$ to $G'$.
However, there is a continuous injection $i : \cC_G \to \cC_{G'}$.
\end{proof}

\begin{proposition}\label{pr:hom_cont}
For graphs $G$ and $G'$ such that there is a w-homomorphism $f : G \to G'$, 
there exists a continuous map $\hf$ from $\cC_G$ to $\cC_{G'}$.
Hence for graphs $G$ and $G'$ such that there is a homomorphism $f : G \to G'$, there exists a continuous map $\hf$ from $\cC_G$ to $\cC_{G'}$.

We say that $\hf$ is (w-)homomorphism type of the (w-)homomorphism $f$.
\end{proposition}

\begin{proof}
By considering Proposition \ref{pr:Hom_vs_wHom}, we will prove the w-homomorphism case.
Let $G=(V, E)$ and $G'=(V', E')$.
We show that there is a map $\hf:\cC_{G} \to \cC_{G'}$ such that 
for any pair of connected subgraphs $(H_1, H_2)$ of $G$ with $H_1 \subset H_2$ we have $\hf(H_1)\subset \hf(H_2)$.
Then, it is continuous due to Lemma \ref{lm:poset_cont}.

As mentioned in Remark \ref{rmk:Hom_wHome}, we consider the graph $G^{\prime\circ}=(V', E^{\prime\circ})$ with loops of $(v',v')\in  E^{\prime\circ}$ for every $v'\in V'$, i.e., $E^{\prime\circ}:= E'\cup\{(v', v')\ | v' \in V'\}$, 
and the projection  $\varpi$ $G^{\prime\circ}\to G'$.
Such a map can be constructed as follows.
Fix $H\subset G$ and consider the set of vertices composing to $H$, denoted by $V(H)$.
From the definition of w-homomorphism, for every $v \in V(H)$, 
we find $f(v) \in V'$ such that $(f(v), f(v'))\in E^{\prime\circ}$ whenever  $(v, v') \in E(H)$, where $E(H)$ is the set of edges of $H$.
We can put a map $g : E(H) \to E^{\prime\circ}$ by $g(e=(v, v'))=(f(v), f(v'))$ and $\hf:=(f(V(H)),\varpi\circ g(E(H)))$.
Clearly, $\hf(H)\in \cC_{G'}$.
Let a pair of connected subgraphs $(H_1, H_2)$ of $G$ satisfying $H_1 \subset H_2$.
It is obvious that $f(V(H_1))\subset f(V(H_2)))$ and $g(E(H_1)) \subset g(E(H_2))$, and thus $\hf(H_1) \subset \hf(H_2)$.
\end{proof}

\begin{table}[htb]
\begin{center}
\caption{Examples of Homomorphism, Hom, w-Homomorphism, w-Hom, and Continuous map, CMap}\label{1tb:Table1}
  \begin{tabular}{|c|c|c|c|}
\hline
$G$ and $G'$ & Hom$(G,G')$ & w-Hom$(G,G')$ &CMap$(C_G, C_{G'})$ \\
\hline
\hline
 \begin{minipage}{0.35\textwidth}
\begin{center}
      \includegraphics[height=0.08\textheight]{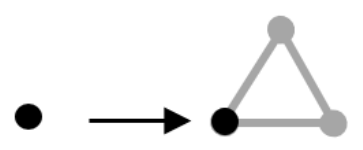}
\end{center}
    \end{minipage}
& exists & exists & exists\\
\hline
 \begin{minipage}{0.35\textwidth}
\begin{center}
      \includegraphics[height=0.08\textheight]{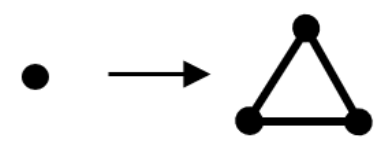}
\end{center}
    \end{minipage}
& not exists & not exists & exists\\
\hline
 \begin{minipage}{0.35\textwidth}
\begin{center}
      \includegraphics[height=0.08\textheight]{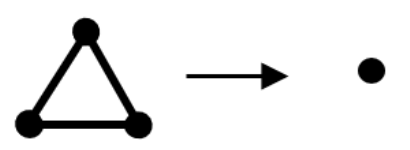}
\end{center}
    \end{minipage}
& not exists & exists & exists\\
\hline
 \begin{minipage}{0.35\textwidth}
\begin{center}
      \includegraphics[height=0.08\textheight]{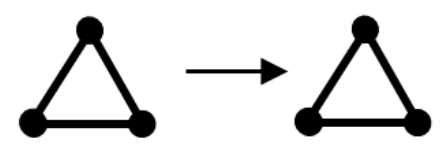}
\end{center}
    \end{minipage}
& exists & exists & exists\\
\hline
 \begin{minipage}{0.35\textwidth}
\begin{center}
      \includegraphics[height=0.08\textheight]{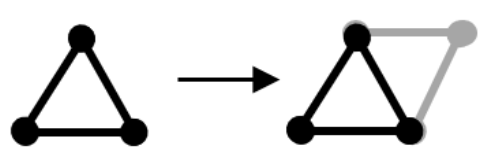}
\end{center}
    \end{minipage}
& exists & exists & exists\\
\hline
 \begin{minipage}{0.35\textwidth}
\begin{center}
      \includegraphics[height=0.08\textheight]{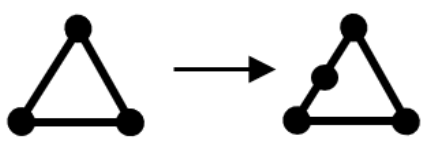}
\end{center}
    \end{minipage}
& not exists & not exists & exists\\
\hline
\hline
  \end{tabular}
\end{center}
\end{table}

\begin{corollary}\label{cj:hom_cont}
For graphs $G=(V,E)$ and $G'=(V', E')$ such that there is a surjective (injective) w-homomorphism $f : G \to G'$, there exists a surjective (injective) continuous $\hf$ map from $\cC_G$ to $\cC_{G'}$.
\end{corollary}

\begin{proof}
We use the notations in the proof in Proposition \ref{pr:hom_cont}.
Let $f$ be surjective, i.e., $f: V \to V'$ is surjective and there is a surjective map $E \to E^{\prime\circ}$ denoted by $g$ such that $g(v,v')=(f(v), f(v'))$.
$C_{G'}$ consists of all subgraphs consisting of subsets of $V'$ and $E'$, whereas $C_{G}$ does of all subgraphs consisting of subsets of $V$ and $E$.
Hence we can find a surjective continuous map $\hf: \cC_G\to \cC_{G'}$.
Similarly, we find an injective continuous map $i: \cC_G\to \cC_{G'}$ for an injective homomorphism $i : G \to G'$.
\end{proof}

\begin{remark}\label{rm:hom_cont*}
{{
Though we are dealing only with the collection $\cC_G$ of connected subgraphs of a graph $G$, if we employ the collection $\widehat{\cC}_G$ of subgraphs of a graph $G$ and the topology of $\widehat{\cC}_G$, then we can have more flexible treatment of continuous map;
for graphs $G=(V,E)$ and $G'=(V', E')$ such that there is a surjective w-homomorphism $f : G \to G'$, there might exist a continuous $\widehat{f}^*$ map from $\widehat{\cC}_{G'}$ to $\widehat{\cC}_{G}$.
}}
\end{remark}

Then the following is obvious.

\begin{lemma}\label{lm:subGcontIII}
If and only if $\cC_G$ and  $\cC_{G'}$ are homeomorphic, $G$ and $G'$ are isomorphic.
\end{lemma}

\begin{remark}
{\rm{
We note that the continuous map  has flexibility rather than the concept of subgraphs and homomorphism as in Table \ref{1tb:Table1}.

}}
\end{remark}

\begin{lemma}
For any finite graphs $G$ and $H$, there exists a graph $G'$ such that
$f: \cC_G \to \cC_{G'}$ and $h: \cC_H \to \cC_{G'}$ are continuous injections.
\end{lemma}

\begin{proof}
We choose that $G'$ has subgraphs $G$ and $H$, and then we obtain the injections.
\end{proof}

This lemma shows that sufficiently large order complete graphs $G_c$ can be a target of the continuous injections from smaller order graphs $H$,
 i.e., $f:\cC_H \to \cC_{G_c}$ is a continuous injection if $|H|<|G_c|$.

\begin{figure}
\begin{center}
\includegraphics[width=0.48\hsize]{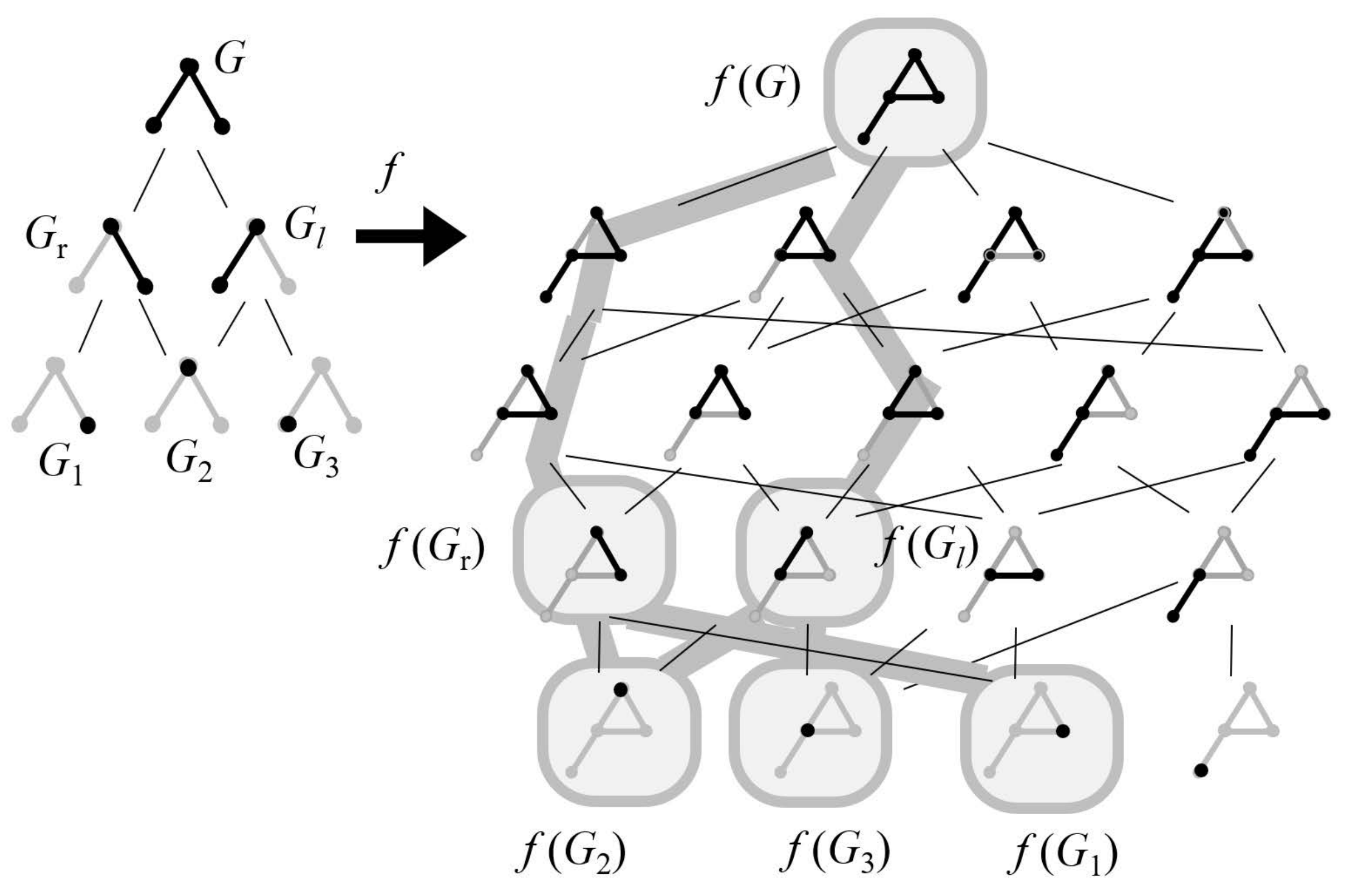}
\hskip 0.03\hsize
\includegraphics[width=0.48\hsize]{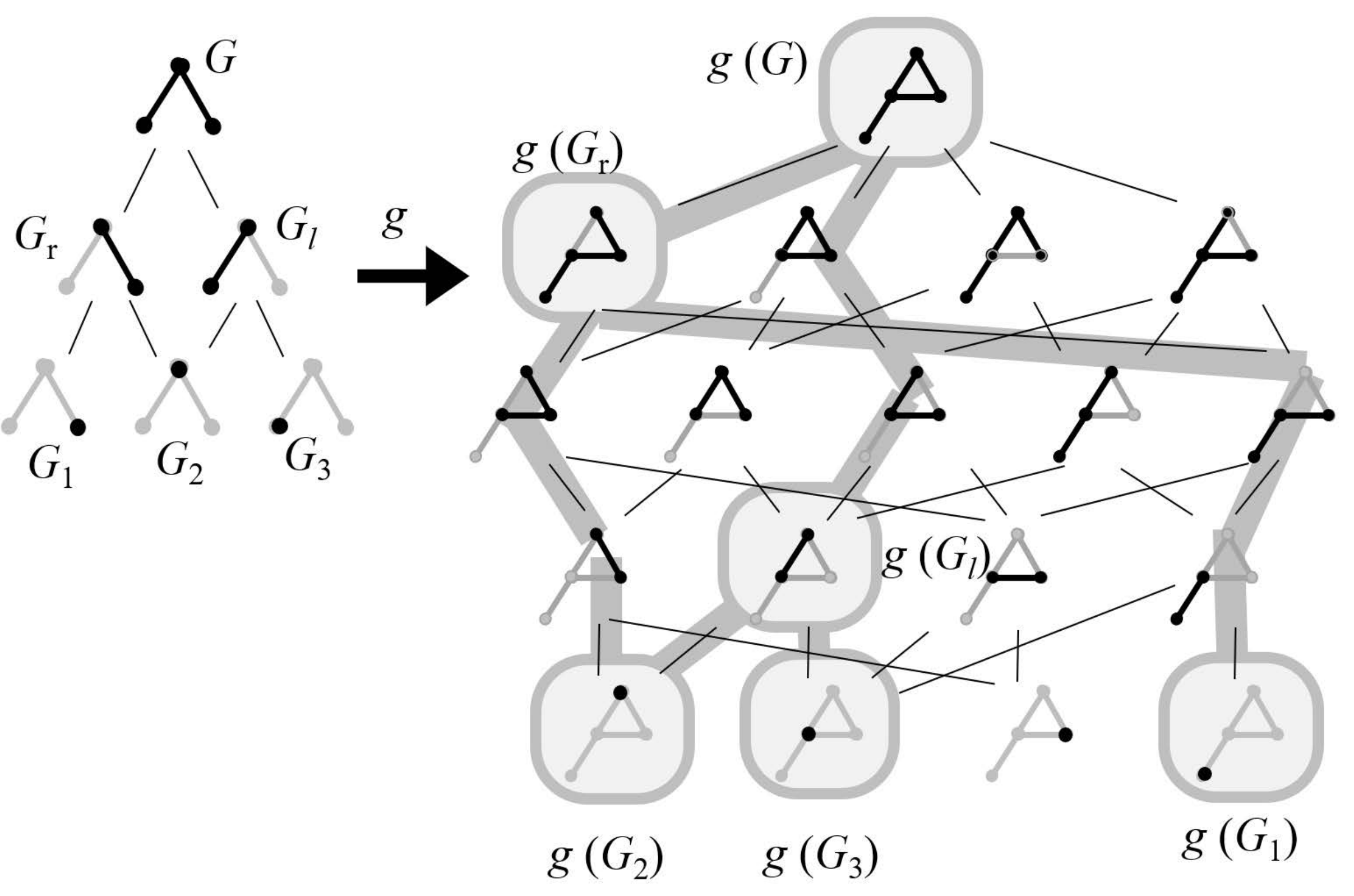}

(a) \hskip 0.48\hsize (b)

\includegraphics[width=0.48\hsize]{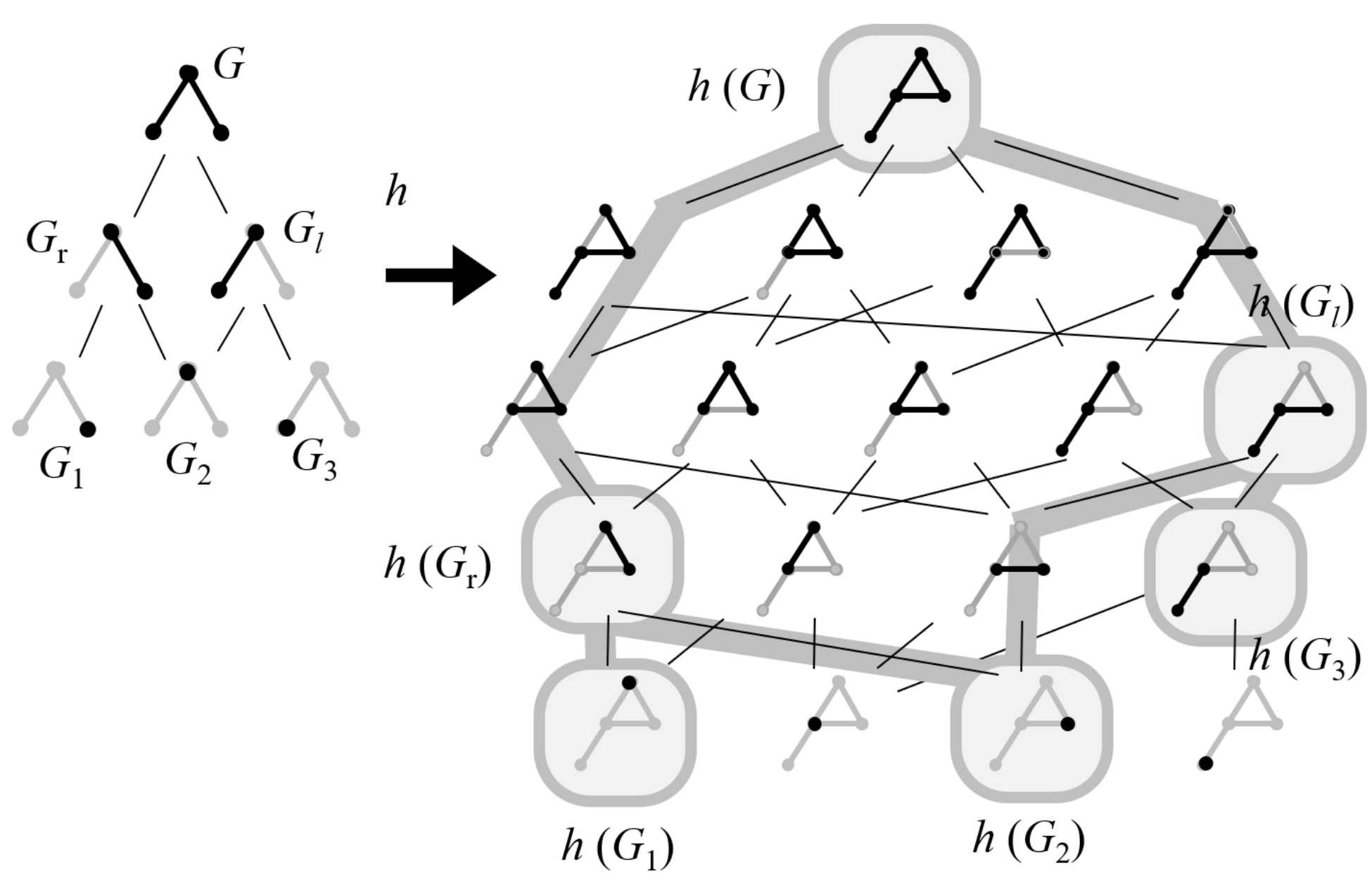}
\hskip 0.03\hsize
\includegraphics[width=0.48\hsize]{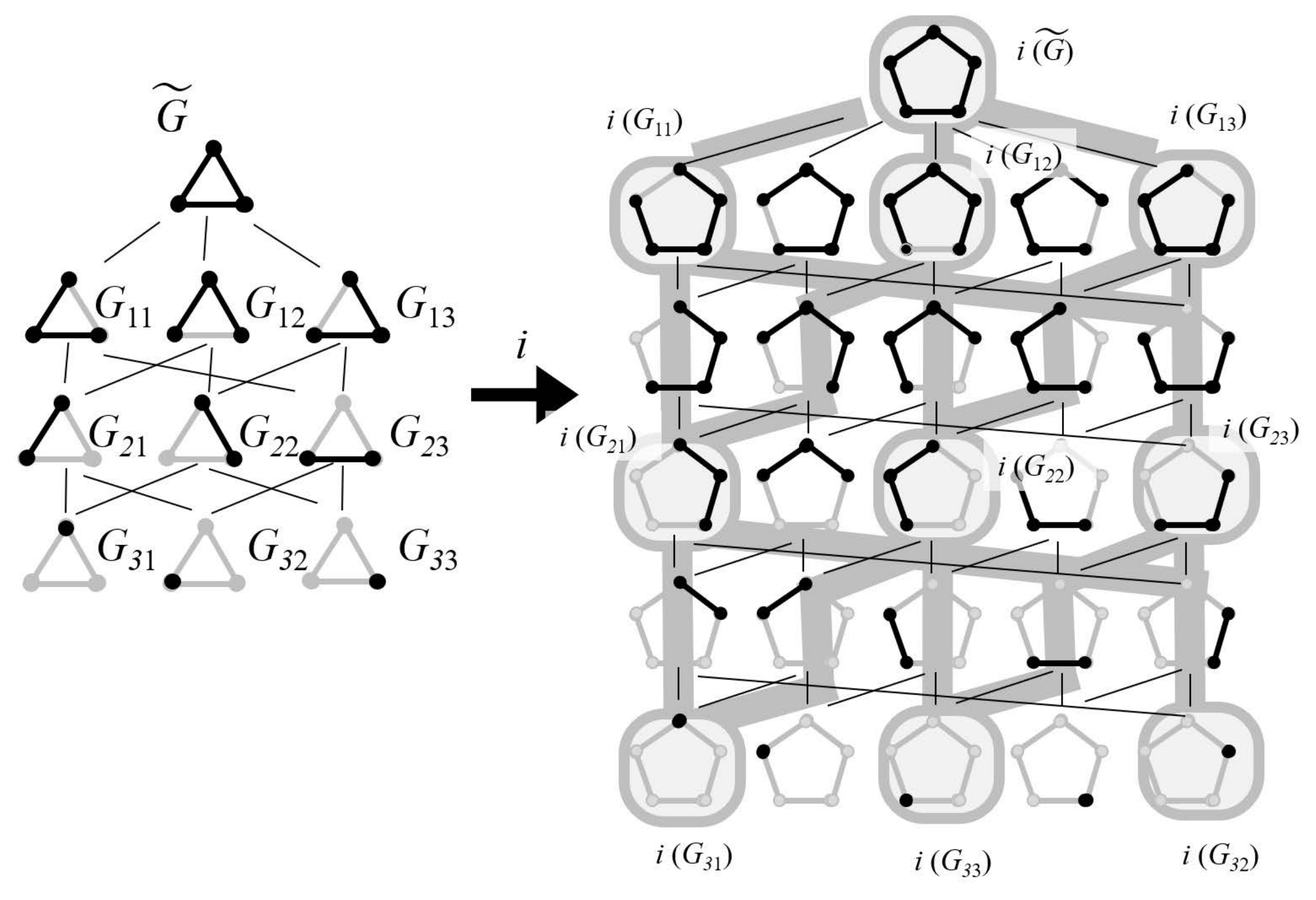}

(c) \hskip 0.48\hsize (d)

\end{center}

\caption{
Examples of continuous injections are illustrated.
}\label{fg:SubGrap2}
\end{figure}

\begin{figure}
\begin{center}
\includegraphics[width=0.48\hsize]{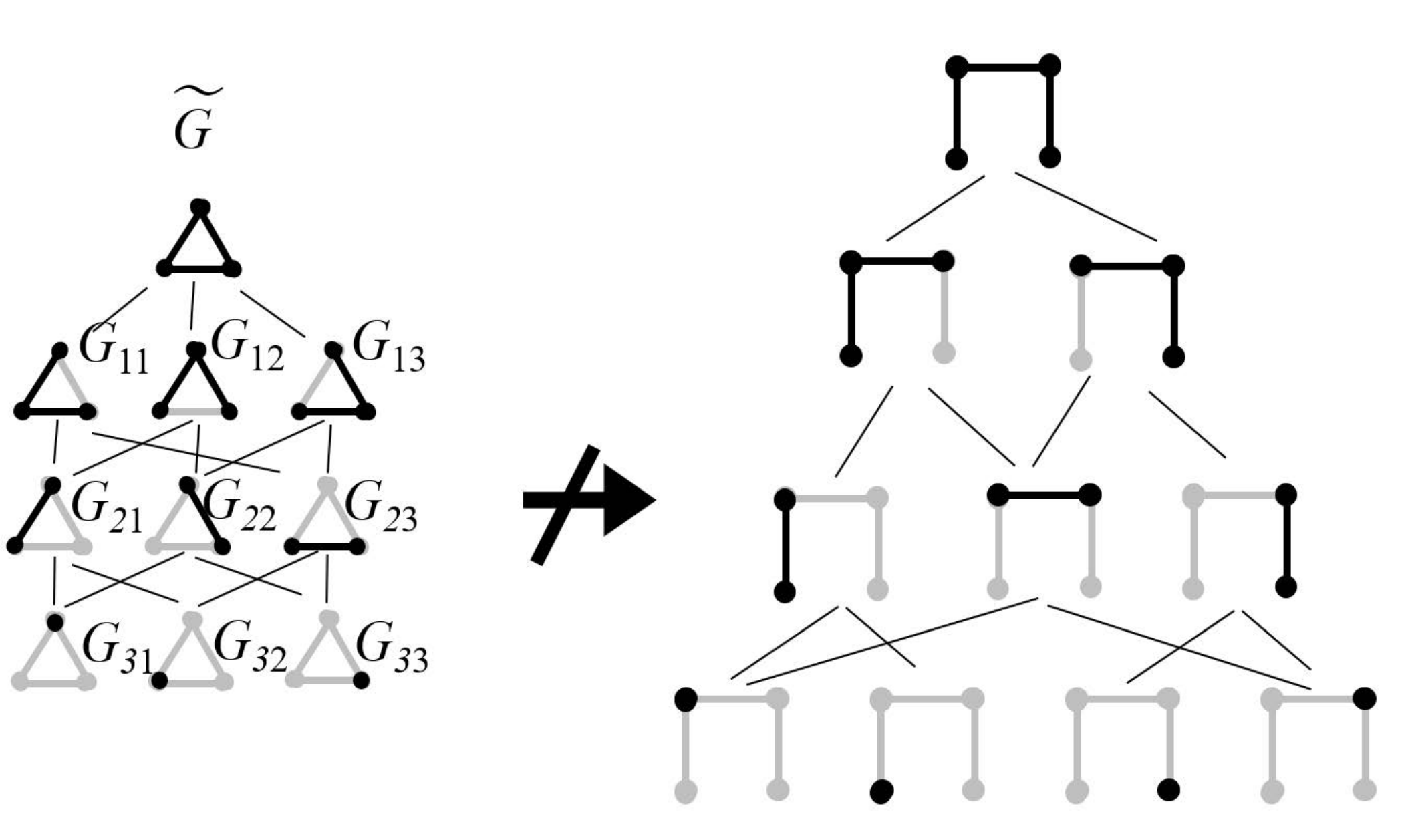}

\end{center}

\caption{
An example which we cannot define an embedding.
}\label{fg:SubGrap3}
\end{figure}

\begin{figure}
\begin{center}
\includegraphics[width=0.80\hsize]{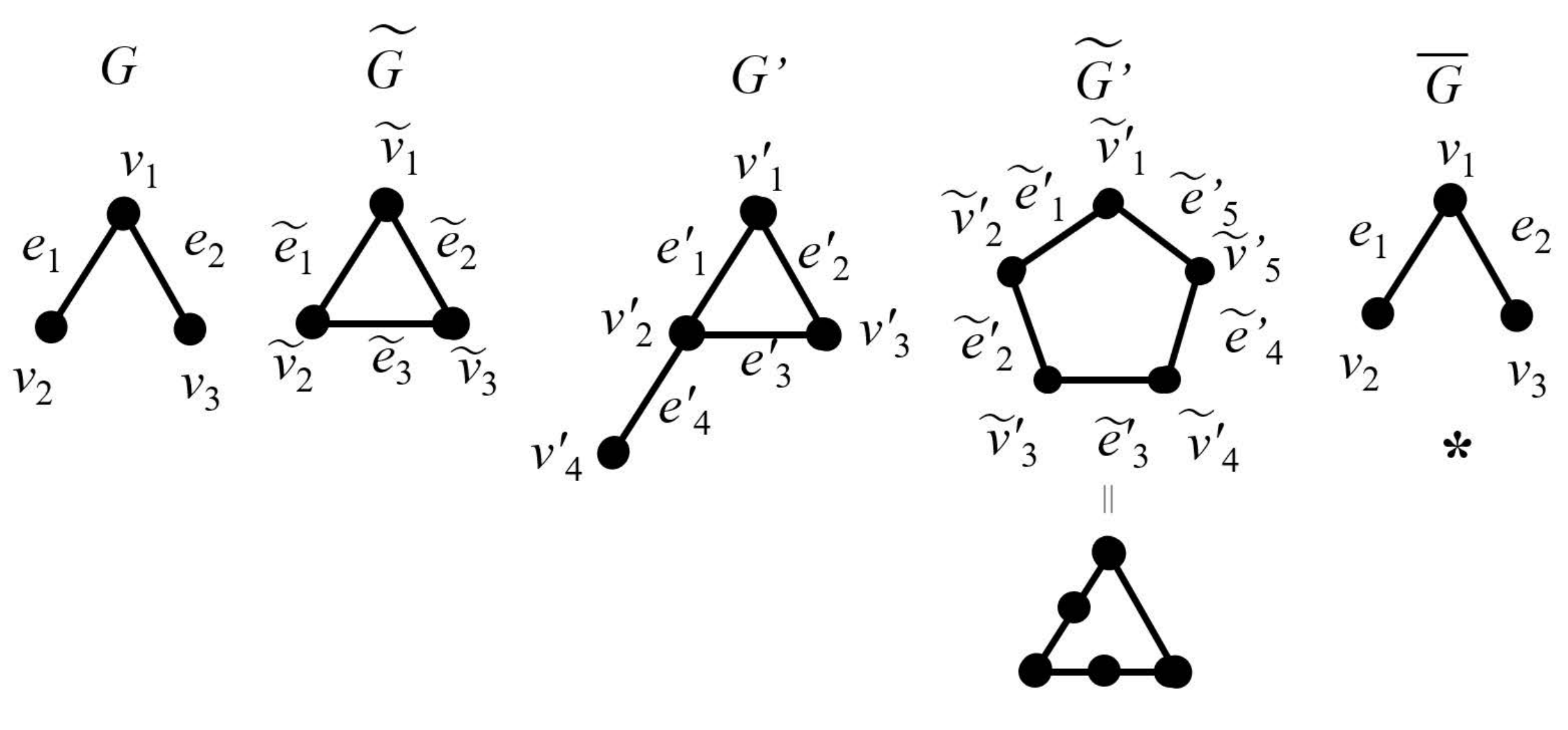}

\end{center}

\caption{
The graphs for the examples of the continuum injections.
}\label{fg:IndGrap1}
\end{figure}

\subsubsection{Examples of the continuous maps}\label{exmp:3.2.1}
We illustrate the examples of the continuous maps in Figure \ref{fg:SubGrap2}.
Futher using Figure \ref{fg:IndGrap1}, we list up the subgraphs and continuous maps as follows.

\begin{enumerate}

\item $G=(V, E)=(\{v_1, v_2, v_3\}, \{e_1, e_2\})$ has the subgraphs,
$G_r = (\{v_1, v_3\}, \{e_2\})$,
$G_\ell = (\{v_1, v_2\}, \{e_1\})$,
$G_1 = (\{v_3\}, \emptyset)$,
$G_2 = (\{v_1\}, \emptyset)$, and 
$G_3 = (\{v_2\}, \emptyset)$.

\begin{enumerate}

\item $f(G_r) = (\{v_1', v_3'\}, \{e_2'\})$, 
$f(G_\ell) = (\{v_1', v_2'\}, \{e_1'\})$, 
$f(G_1) = (\{v_3'\}, \emptyset)$,
$f(G_2) = (\{v_1'\}, \emptyset)$,
and
$f(G_3) = (\{v_2'\}, \emptyset)$.

\item $g(G_r) = (\{v_1', v_2', v_3', v_4'\}, \{e_2', e_3', e_4'\})$, 
$g(G_\ell) = (\{v_1', v_2'\}, \{e_1'\})$, 
$g(G_1) = (\{v_4'\}, \emptyset)$,
$g(G_2) = (\{v_1'\}, \emptyset)$,
and
$g(G_3) = (\{v_2'\}, \emptyset)$.

\item $h(G_r) = (\{v_1', v_3'\}, \{e_2'\})$, 
$h(G_\ell) = (\{v_2', v_3', v_4'\}, \{e_3', e_4'\})$, 
$h(G_1) = (\{v_1'\}, \emptyset)$,
$h(G_2) = (\{v_3'\}, \emptyset)$,
and
$h(G_3) = (\{v_2',v_4'\},\{e_4'\})$.

\end{enumerate}

\item $\tG=(\tV, \tE)=(\tv_1, \tv_2, \tv_3\}, \{\te_1, \te_2, \te_3\})$ has the subgraphs,
$\tG_{21} = (\{\tv_1, \tv_2\}, \{\te_1\})$,
$\tG_{22} = (\{\tv_1, \tv_3\}, \{\te_2\})$,
$\tG_{23} = (\{\tv_2, \tv_3\}, \{\te_3\})$,
$\tG_{31} = (\{\tv_1\}, \emptyset)$,
$\tG_{32} = (\{\tv_2\}, \emptyset)$, and 
$\tG_{33} = (\{\tv_3\}, \emptyset)$.

\begin{enumerate}

\item 
$i(\tG_{21}) = (\{\tv_1', \tv_5', \tv_4'\}, \{\te_4', \te_5'\})$, 
$i(\tG_{22}) = (\{\tv_1', \tv_2', \tv_3'\}, \{\te_1', \te_2'\})$, 
$i(\tG_{23}) = (\{\tv_3', \tv_4', \tv_5'\}, \{\te_3', \te_4'\})$, 
$i(\tG_{31}) = (\{\tv_1'\}, \emptyset)$,
$i(\tG_{32}) = (\{\tv_5'\}, \emptyset)$,
and
$i(\tG_{33}) = (\{\tv_3'\}, \emptyset)$.

\end{enumerate}

\end{enumerate}

For the continuous injection of $f: \cC_G \to \cC_{G'}$, we may find an induced map $f^*: G' \to G$.
In order to define it, we introduce a extension of a graph illustrated in Figure \ref{fg:IndGrap1}.

\begin{definition}
For a connected graph $G$ and a generic element $*$, we define the pointed graph $(G, *)$ which is denoted by $\BG$.
\end{definition}

The following is obvious.

\begin{proposition}
Assume that $G=(V, E)$, where $V=\{v_i\}_{i=1, \ldots, n_V}$ and $E=\{e_i\}_{i=1, \ldots, n_E}$, and  $G'=(V', E')$.
Let $t(e)$ be the vertices as ends of an edge $e$ such that $(t(e), e) \in \cC_G$, and $t_0(e)$ be one of $t(e)$ whose index is minimal one of $t(e)$.
For the continuous injection of $f: \cC_G \to \cC_{G'}$, there is a map $f^*$ from $\BG'$ to $\BG$ satisfying the following

\begin{enumerate}
\item For an element $f(\{v_i\}, \emptyset)=(V_{v_i}', E_{v_i}')$, $(i=1, \ldots, n_V)$, $v_j' \in V_i'$, and $e_\ell' \in E_{i}'$, $f^*(v_j') := v_i$, and $f^*(e_\ell'):=v_i$, let $V'_V:=\cup_{i=1, \ldots, n_V} V_{v_i}'$.

\item 
For an element $f(t(e_i), e_i)=(V_{e_i}', E_{e_i}')$, $(i=1, \ldots, n_E)$, $v_j' \in V_{e_i}'\setminus V'_V$, and $e_\ell' \in E_{e_i}'\setminus(\cup_{j<i}E_{e_j}')$, $f^*(v_j') := t_0(e_i)$ and $f^*(e_\ell'):=e_i$, let $V'_E:=\cup_{i=1, \ldots, n_V} V_{e_i}'$ and $E'_E:=\cup_{i=1, \ldots, n_V} E_{e_i}'$.

\item 
For an element $v' \in V' \setminus (V'_V\cup V'_E)$, let $f^*(v') = *$ and
for an element $e' \in E' \setminus E'_E$, let $f^*(e') = *$.

\end{enumerate}

\end{proposition}

\begin{proof}
We note that every graph $G$ is determined by the set of vertices $V$, which corresponds to the collection of its subgraphs given by single-order graphs, and the set of edges $E$, which corresponds to the collection of its subgraphs with a single edge.
Thus we consider the image of $f^*$ by considering two cases.
(1) means that we consider the preimage of $V$ for $f^*$  by considering the correspondence $f$, whereas (2) shows the preimage of $E$.
However, since multiple preimages exist for elements of $V$ and $E$, we avoid multiple definitions.
Further, since there are elements in $G'$ that have no correspondence from $G$ and its subgraph, we define the map from these to $*$ in $G$.
Then we define a map from $\BG'$ to $\BG$.
\end{proof}

\subsubsection{Examples of the induced maps}

For the examples in Subsubsection \ref{exmp:3.2.1} associated with Figure \ref{fg:SubGrap2}, we show the examples of the induced maps as follows.

\begin{enumerate}

\item $f^*(v_1') =v_1$, $f^*(v_2') = v_2$, $f^*(v_3')= v_3$, $f^*(v_4')=*$, 
 $f^*(e_1') =e_1$, $f^*(e_2') = e_2$, $f^*(e_3')= *$, $f^*(e_4')=*$, 
$f^*(*) =*$.

\item 
$g^*(v_1') =v_1$, $g^*(v_2') = v_2$, $g^*(v_3')= v_1$, $g^*(v_4')=v_3$, 
 $g^*(e_1') =e_1$, $g^*(e_2') = e_2$, $g^*(e_3')= e_2$,   $g^*(e_4')=e_2$, 
$g^*(*) =*$.

\item 
$h^*(v_1') =v_3$, $h^*(v_2') = v_2$, $h^*(v_3')= v_1$, $h^*(v_4')=v_2$, 
 $h^*(e_1') =*$, $h^*(e_2') = e_2$, $h^*(e_3')= e_1$, $h^*(e_4')=v_2$, 
$h^*(*) =*$.

\item 
$i^*(\tv_1') =\tv_1$, $i^*(\tv_2') = \tv_1$, 
$i^*(\tv_3')= \tv_3$, $i^*(\tv_4')=\tv_2$, 
$i^*(\tv_5')=\tv_2$,
 $i^*(\te_1') =\te_2$,
 $i^*(\te_2') = \te_2$, 
$i^*(\te_3')= \te_3$, 
$i^*(\te_4')=\te_1$, 
$i^*(\te_5')=\te_1$, 
$i^*(*) =*$.

\end{enumerate}

\begin{definition}
For a continuous injection of $f: \cC_G \to \cC_{G'}$, and its induced  map $f^*$ from $\BG'\to \BG$, we say that the correspondence $f^\#: G' \to \BG$ is injective continuous morphism induced from $f$.
\end{definition}

\section{Conway's law}

Let $G_S$ be the graph of the geometric structure of a software system $S$
and $G_O$  the graph of the geometric structure of the organization $O$ producing $S$.

To simplify the discussion, let us make the following assumptions.
The graph of the software system $S$ consists of $G_S=(V_S, E_S)$, where the set of nodes $V_S$ is a set of modules in the software $S$ and the set of edges $E_S$ represents the relations among $V_S$.
$G_O=(V_O, E_O)$ is given by the set of nodes $V_O$ whose element is a team (or a person) in the organization $O$ and the set of edges $E_O$ which represents the relations among $V_O$.

\subsection{The first step: original Conway's law with task graph}
The following is the original Conway's law \cite{Conway} by adding the task graphs.

\begin{law}\label{law1}
Let $G_S=(V_S, E_S)$ be the finite connected graph of the software system $S$
and $G_O=(V_O, E_O)$ the finite connected graph of the organization $O$ producing $S$.
There exists a homomorphism $q: G_S \to G_O$.

More precisely, 
there exists a graph $G_T=(V_T, E_T)$ of the structure of the task of the software system $G_S$ such that there exists a surjective homomorphism $p: G_S \to G_T$ and injective homomorphism $i: G_T \to G_O$ satisfying $q = i \circ p$, i.e., \begin{equation}
\xymatrix{ G_S \ar[d]^{p}\ar[dr]^{q}&\\
           G_T \ar[r]^{i}& G_O }
\end{equation}
\end{law}

\begin{definition}
For $G_S$, $G_O$, and $G_T$ in Law \ref{law1}, $(G_S, G_O)$ and $(G_S, G_O, G_T)$ are called a Conway doublet and a Conway triplet respectively.
Further, we say that $q$ is a Conway correspondence and $(p,i)$ is the Conway decomposition of $q$.
\end{definition}

\begin{remark}\label{rk:ConwayLaw1}
{\rm{
As mentioned in Introduction, Conway implicitly stated the following:
\begin{enumerate}
\item If $(G_S,G_O)$ is not Conway's doublet, instead of $G_S$, a software system with a truncated graph $G_S'$ of $G_S$ such that there is homomorphism $G_S'$ to $G_O$ is produced.

\item If $q$ itself is injective, $p$ is the isomorphism (though it is not usual.)

\item 
The task graph represents that for a vertex $v$ of $f(G_S)$ corresponding a person or a team, it is responsible to the tasks $p^{-1}(\{i^{-1}(v)\})=q^{-1}(\{v\})$, as a fiber of $p: G_S \to G_T$.
The projection $p$ shows the binding the responsible jobs to the modules or subsystems as a task.

\item 
Conway implied that $G_T$ and $i(G_T)$ should be simple to avoid extra communication.
In order to avoid unnecessary communication, 1) $i$ must be isomorphic to the subgraph consisting of $q(V_S)$ and the set of edges which is a subset of $E_O$ whose connected $q(V_S)$, and 2) $G_T$ must have a simple structure.
As shown in figure \ref{fg:GS_GT}, there are several surjective homomorphisms.
Depending on the morphisms, the complexity of $G_T$ differs. 
Since the complex structure requires extra communication between teams, $G_T$ should be simple.

\end{enumerate}
}}
\end{remark}

\begin{figure}
\begin{center}
\includegraphics[width=0.4\hsize]{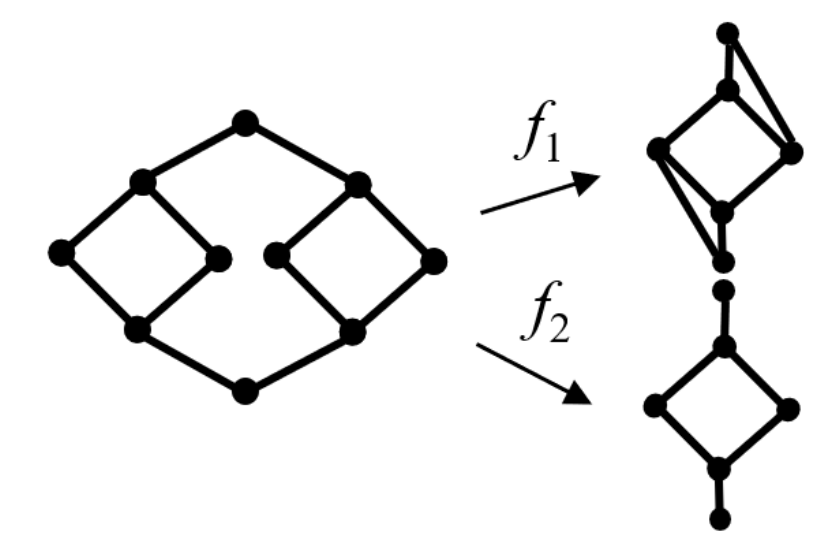}

\end{center}

\caption{
There are two surjective homomorphisms $f_1$ and $f_2$
}\label{fg:GS_GT}
\end{figure}

\subsection{The second step: Conway's law with w-homomorphism}

Second we weaken these by w-homomorphisms noting Proposition \ref{pr:Hom_vs_wHom}.

\begin{law}\label{law2}
For the situation in Law \ref{law1}, there exists a w-homomorphism $q: G_S \to G_O$.

More precisely, 
there exists a graph $G_T=(V_T, E_T)$ of the structure of the task of the software system $G_S$ such that there exists a surjective w-homomorphism $p: G_S \to G_T$ and injective homomorphism $i: G_T \to G_O$ satisfying $q = i \circ p$, i.e., \begin{equation}
\xymatrix{ G_S \ar[d]^{p}\ar[dr]^{q}&\\
           G_T \ar[r]^{i}& G_O }
\end{equation}
\end{law}

\begin{definition}
For $G_S$, $G_O$, and $G_T$ in  Law \ref{law2}, $(G_S, G_O)$ and $(G_S, G_O, G_T)$ are called a w-Conway doublet and a w-Conway triplet respectively.
Further, we say that $q$ is a w-Conway correspondence and $(p,i)$ is the w-Conway decomposition of $q$.
\end{definition}

\subsection{The third step: Conway's law under graph topology}

As Conway wrote {\lq\lq}The structures of large systems tend to disintegrate during development, qualitatively more so than with small systems{\rq\rq} \cite{Conway}, we should design the organization. 

In Conway's modeling, he stated that the geometrical structure has a hierarchical structure as mentioned later (see Examples \ref{ssec:exmp_Conway} \ref{itm:Hir7}).
Thus, when we compare the $G_S$ and $G_O$, it must be necessary that both abstraction levels are the same.
However, it might be difficult.
We should describe the hierarchical structure of the system and the organization in a unified manner.
The homomorphism is not a proper tool to express it, as mentioned in Remark \ref{rmk:Hom_wHome2}.
Although w-homomorphism could express it, graph topology is a more natural tool to describe it completely.

Thus, we rephrase them by using the graph topology.

\begin{law}\label{law3}
For the situation in Law \ref{law1}, there exists a continuous map $\hq: \cC_{G_S} \to \cC_{G_O}$.

More precisely, 
there exists a graph $G_T=(V_T, E_T)$ of the structure of the task of the software system $G_S$ such that there exist a continuous surjection $\hp: \cC_{G_S} \to \cC_{G_T}$ with w-homomorphism $p: G_S \to G_T$ and a continuous injection $\hati: \cC_{G_T} \to \cC_{G_O}$ satisfying $\hq = \hati \circ \hp$, i.e.,  
\begin{equation}
\xymatrix{ \cC_{G_S} \ar[d]^{\hp}\ar[dr]^{\hq}&\\
           \cC_{G_T} \ar[r]^{\hati}& \cC_{G_O} }
\end{equation}
\end{law}

\begin{definition}
For $G_S$, $G_O$, and $G_T$ in  Law \ref{law3}, $(G_S, G_O)$ and $(G_S, G_O, G_T)$ are called a t-Conway doublet and a t-Conway triplet respectively.
Further, we say that $\hq$ is a t-Conway correspondence and $(\hp,\hati)$ is the t-Conway decomposition of $\hq$.

Let $p$ be a surjective w-homomorphism $p$ for $\hp$.
We refer to the correspondence $\hq^\#: G_O\to \BG_S$ given by $\hq^\#(v):=p^{-1}(\{\hati^{*-1}(v)\}) \subset V_S$ and $\hq^\#(e):=p^{-1}(\{\hati^{*-1}(e)\})$ for $(v,e) \in (V_O, E_O)$ as a t-Conway morphism.
Here we let $p^{-1}(\{*\})$ be $* \in \BG_S$.

If there is no continuous map from $\cC_{G_S}$ to $\cC_{G_O}$, then we say that $G_S$ is a more complicated graph than $G_O$.
\end{definition}

\begin{proposition}
If $(G_S, G_O)$ and $(G_S, G_O, G_T)$ are a Conway doublet and a Conway triplet respectively, they are t-Conway (w-Conway) doublet and t-Conway (w-Conway) triplet.
\end{proposition}

\begin{proof}
It is obvious due to Proposition \ref{cj:hom_cont}.
\end{proof}

Let us consider the case that $G_S$ is a more complicated graph than $G_O$.
Then we find a truncated graph $G$ of $G_S$.
\begin{proposition}\label{pr:GOtoGSbytrunc}
For given $G_O$ and $G_S$, there is a truncated graph $G_S'$ of $G_S$ such that
there is a continuous injection $\hq: \cC_{G_S'}\to \cC_{G_O}$.
\end{proposition}

\begin{proof}
Let us consider an edgeless graph $G_S' = (\{v\}, \emptyset)$.
Lemma \ref{lm:singleVer} shows it.
\end{proof}

\subsection{Examples}\label{ssec:exmp_Conway}
Let us consider examples:

\begin{enumerate}

\item Let $G_O=(\{v\}, \emptyset)$.
Then for any non-empty finite connected graph $G_S$, there is a w-homomorphism $f: G_S \to G_O$, and thus $(G_S, G_O)$ is w-Conway doublet, and thus t-Conway doublet.

It means that a super-person or a very small group can make any complicated system, as mentioned in Introduction.
(It maybe recall the birth of Linux.)
Then it is decomposed by $G_T=(\{v\}, \emptyset)$.

Though Conway forbade internal communication in his model and this phenomenon cannot be described in Law \ref{law1},  we consider that this should be also described. 
Thus we believe that Law \ref{law2} and Law \ref{law3} are more natural than Law \ref{law1}.

\smallskip

However, this may be a not good example in software development because it cannot be scaled up.
It comes from the fact that the task graph does not reflect the geometric structure of $G_S$.

\item\label{item:G_S'} Assume that there is neither continuous map from $\hq: \cC_{G_S} \to \cC_{G_O}$ nor w-homomorphism $q: G_S \to G_O$ except such trivial maps in $\hq : \cC_{G_S}\to \cC_{(\{v\}, \emptyset)} \subset \cC_{G_O}$ or 
$q : G_S\to (\{v\}, \emptyset) \subset G_O$, and so on.
Even for such a case, we find a software structure $G_S'$ due to Proposition \ref{pr:GOtoGSbytrunc}. 

This corresponds to Remark \ref{rk:ConwayLaw1} 1 of law \ref{law1};
The restriction to homomorphism is much simpler than such exceptions of trivial maps in Law \ref{law2} or \ref{law3}.

\item 
Suppose we neglect cognitive load and avoidance of unnecessary communication. Then, if $G_O$ is a complete graph (i.e., an organization with dense communication) of sufficient size, any $G_S$ can be injective.

This maybe a model of gather-discussion (informal and open discussion/meeting in Japanese company.)

\item 
Suppose $G_S$ consists of several subgraphs $H_i$ $(i=1, \ldots, \ell)$ which are isomorphic to $H$, and there is a homomorphism $q: G_S\to G_O$.
Then there may exist $q$ such that $q(H_i) = q(H_1)$ for $i=2 \ldots, \ell$ as in Figure \ref{fg:Hom_vs_wHom} (a).
We can find a task graph $G_T$ such that there exists a surjective (w-)homomorphism $p: G_S \to G_T$ and an injective homomorphism $i: G_T \to G_O$.
A team or person corresponding to a vertex $v_O$ of $q(H_1)$ is responsible for the set of different modules, $i^{-1}(v_O)$.

As mentioned in Remark \ref{rk:ConwayLaw1}, how we construct the homomorphism $p$ is important if we consider cognitive load and avoidance of unnecessary communication.
The simple construction of $G_T$ is needed.

This view is much more important when we consider the w-homomorphism in Law \ref{law2}, in which we allow the internal communication.

\smallskip

We continue to comment on the task graphs of Law \ref{law3} because it contains the following cases (\ref{itm:6}) and (\ref{itm:7}).
For a t-Conway triplet $(G_S, G_O, G_T)$ of Law \ref{law3} and its continuous injection $\hati : \cC_{G_T}\to \cC_{G_O}$,
the induced map $\hati^* : \BG_O=(G_O, *) \to \BG_T=(G_T, *)$ shows which team (who) is responsible for which modules in $S$ or which group (who) in $O$ is responsible for which module in $S$, i.e., $\hq^\#(v)=\hp^{-1}(\{\hati^{*-1}(v)\})$.
It is important to note that certain teams are not responsible for a certain software $S$ can also be expressed by $*$ in $\BG_T$; 
The inverse map $\hati^{*-1}(*)\in G_O$ does not contributes the software.

As mentioned above, the task graph represents that a vertex $v$ of $\BG_O \setminus \{\hati^{*-1}(\{*\})\}$ corresponding a person or a team is responsible to the tasks $p^{-1}(\{\hq^*(v)\})$, as a fiber of $p: G_S \to G_T$.

Also note that the image of $\hati(v)$ for $v\in G_T$ may have geometric structure as in the following (\ref{itm:6}) and (\ref{itm:7}).

For cognitive load and avoidance of unnecessary communication, we can express such situations in terms of t-Conway triplet.

\item\label{itm:6}
Assume that there is a homomorphism $f: G_S\to G'$, and a continuous map $\hf: \cC_{G_S}\to \cC_{G'}$ corresponding to $f$ due to Proposition \ref{cj:hom_cont}.
We also assume that $G_O$ is obtained by replacing every vertex $v_i$ in $G'$ with a certain non-trivial graph $H_i$.
Then in general, there is no homomorphism nor w-homomorphism from $G_S$ to $G_O$.
However there exist a continuous map $\hq: \cC_{G_S}\to \cC_{G_O}$ induced from $f$.
The t-Conway state represents how a team responsible for a module of $G_S$ has its internal geometric structure.
It is quite natural in real management.

\item\label{itm:7} 
Assume that there is a homomorphism $f: G_S\to G'$, and a continuous map $\hf: \cC_{G_S}\to \cC_{G'}$ corresponding to $f$ due to Proposition \ref{cj:hom_cont}.
We also assume that $G_O$ is obtained by replacing every edge $e_i$ in $G'$ with a particular non-trivial graph $H_i$ such that the connection preserves.
Then, generally, there is no homomorphism nor w-homomorphism from $G_S$ to $G_O$.
On the other hand, there exist a continuous map $\hq: \cC_{G_S}\to \cC_{G_O}$ induced from $\hf$ and $f$.

Conway considered the edge in the organization $G_O$ to be a coordinator of communication.
In Conway's time it is appropriate.

However, recently the importance of the communication increases, due to secure, knowledge hiding, and the object that should be controlled.
Thus, it could be responsible by a special team.
The t-Conway doublet can express it, though Conway doublet cannot express it.

\item\label{itm:Hir7}
Figure \ref{fg:SubGrap2} shows the continuous maps. 
(Some of them are not homomorphism.)

Then it is obvious that the node in $G_S$ or $G_T$ can correspond to a subgraph of $G_O$ in the t-Conway triplet.
This means that a module or a subsystem can be made responsible by a team or an association of teams that has internal geometric properties.

Conway noted that the geometric structure has hierarchical properties, and thus the nodes in the graphs in \cite{Conway} correspond to objects at a certain level.
In other words, he considered the projective graph sequences for the systems and organizations and the correspondences as mentioned in Remark \ref{rmk:Hom_wHome2}.
Figure \ref{fg:HSseq_HOseq} (a) is the redrawing of the hierarchical sequences of graphs of software systems by using the w-homomorphism $p_i$ of Figure 1 in \cite{Conway}, though the arrows are inverse of the original one.
The larger $i$ of $G_{S,i}$ is the coarser geometric structure of $G_S$, and thus the smaller $i$ of $G_{S,i}$ is the finer geometric structure of $G_S$.
To capture the geometric structure of the system, this hierarchical treatment is necessary.

Conway also stated that such a structure in the system appears in the organization like Figure \ref{fg:HSseq_HOseq} (b).
Among them, he implicitly considered the correspondences by homomorphisms.
(As mentioned in (2), we consider that $q_i$ should be homomorphism or homomorphism-type.)
\begin{equation}
\xymatrix{ 
G_{S}=G_{S,0} \ar[r]^{p_0}\ar[d]^{q_0} & \cdots  \ar[r]^{p_{i-1}}& G_{S,i}  \ar[r]^{p_i}\ar[d]^{q_{i}} &G_{S,i+1} \ar[r]^{p_{i+1}} \ar[d]^{q_{i+1}}&\cdots 
 \ar[r]^{p_t\ \ \ \ }& G_{S,t}=(\{v\},\emptyset) \ar[d]^{q_{t}}\\
G'_{O}=G_{O,0} \ar[r]^{p'_0} & \cdots  \ar[r]^{p'_{i-1}}& G_{O,i}  \ar[r]^{p'_i}& G'_{O,i+1} \ar[r]^{p'_{i+1}} &\cdots  \ar[r]^{p'_t\ \ \ \ } &G'_{O,t}=(\{v'\},\emptyset),
 }
\label{eq:Seq_GS_GO}
\end{equation}
As mentioned in Remark \ref{rmk:Hom_wHome2}, they cannot be expressed only by the homomorphism because these $p_i$ and $p'_i$ in Figure \ref{fg:HSseq_HOseq} are not homomorphism but are w-homomorphisms.

Further, we can consider a relation between $G_{S,i+1}$ and $G_{O, i}$ because it shows how the coarse structure is realized in the finer structure of the organization.
Management of the organization sometimes requires such a view of realization.
However, even if there is a homomorphism $G_{S,i+1}\to G_{O, i+1}$, there cannot exist a natural w-homomorphism (nor homomorphism) $G_{S,i+1}\to G_{O, i}$ though we can express it by the continuous map.
In other words, to handle the hierarchical structure with Conway's correspondence, the t-Conway correspondence is better to express it than the ordinary Conway's treatment as illustrated in Figure \ref{fg:HSseq_HOseq} (a) and (b).

Furthermore Remark \ref{rm:hom_cont*} means that if we employ a variant version of the collection of subgraphs $\widehat{\cC}_G$, and the w-homomorphisms are surjective, we may also find a continuous map from  $\widehat{\cC}_{G_{O, i}}$ to $\widehat{\cC}_{G_{S,j}}$, $j<i$, which shows who is responsible for the created things.

\begin{figure}
\begin{center}
\includegraphics[width=0.65\hsize]{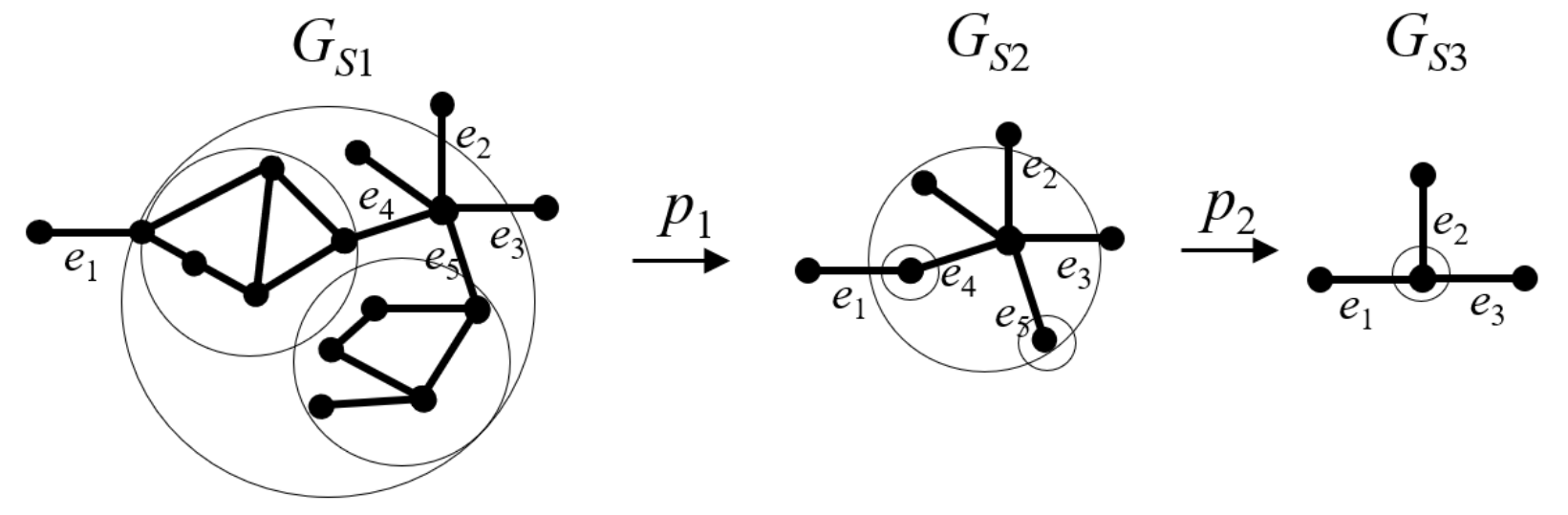}

(a)

\includegraphics[width=0.65\hsize]{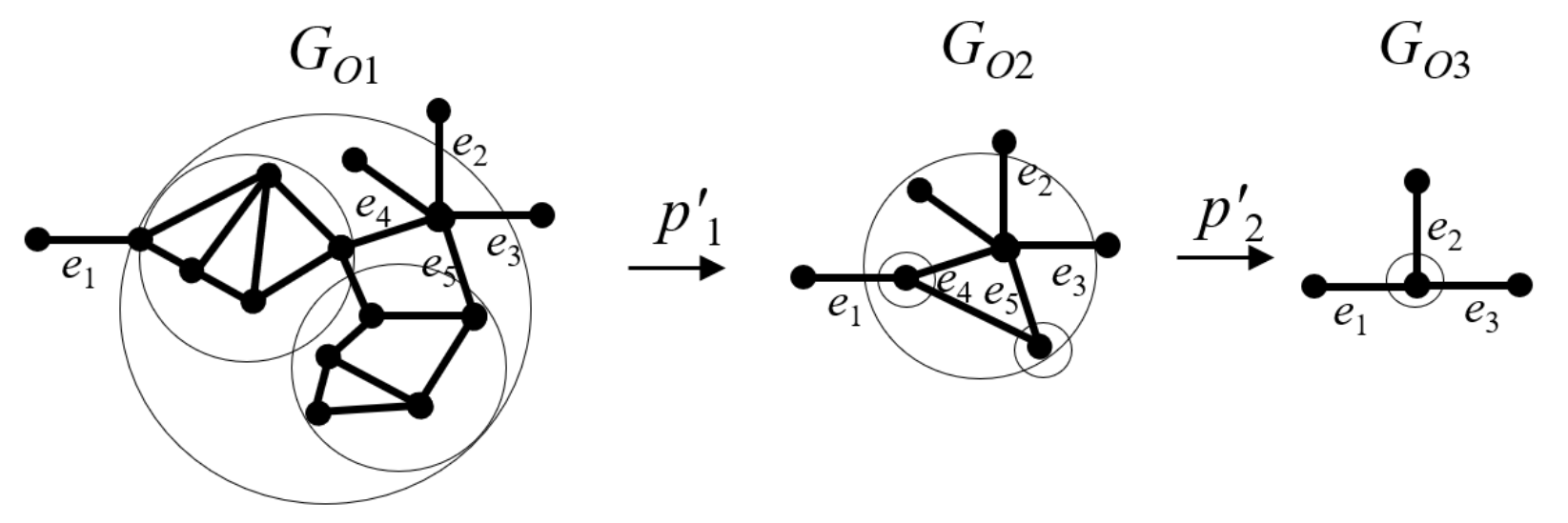}

(b)

\end{center}

\caption{
(a) is the hierarchical sequence of graphs of software systems by w-projection $p_i$ corresponding to figure 1 in \cite{Conway}, and (b) is the hierarchical sequence of organization graphs.
}\label{fg:HSseq_HOseq}
\end{figure}

\item 
By considering (\ref{itm:Hir7}), we comment on (\ref{item:G_S'})  and Remark \ref{rk:ConwayLaw1} (1) again.
Conway wrote 
{\lq\lq}First, the realization by the initial designers that the system will be large, together with certain pressures in their organization, make irresistible the temptation to assign too many people to a design effort. 
Second, application of the conventional wisdom of management to a large design organization causes its communication structure to disintegrate.
Third, the homomorphism insures that the structure of the system will reflect the disintegration which has occurred in the design organization{\rq\rq} \cite{Conway}.
As he also noted, there is a life cycle of a system design; as the system progresses in time and grows in size, the geometric structure of the software system undergoes a phenomenon in which the geometric structure of the organization does not keep up with the changes.

As mentioned in (\ref{item:G_S'}) and Remark \ref{rk:ConwayLaw1} (1), there is a concern that the geometric structure of the software may be changed to fit the organizational structure without being aware of such a deviation from what is required by the geometric structure of the software system.

The geometry of an organization is generally limited in its freedom of choice by the historical background of its establishment, its administrative structure and its financial situation.
Conway stated {\lq\lq}Common management practice places certain numerical constraints on the complexity of the linear graph which represents the administrative structure of a military-style organization. 
Specifically, each individual must have at most one superior and at most approximately seven subordinates.
To the extent that organizational protocol restricts communication along lines of command, the communication structure of an organization will resemble its administrative structure{\rq\rq} \cite{Conway}.

We emphasize that (\ref{eq:Seq_GS_GO}) shows the guideline of how to deal with the geometric structure of the software system mathematically.
We may consider how to handle the geometric structure (more precisely task graphs) for the given geometrical structures of the organization and the software system by using the diagram (\ref{eq:Seq_GS_GO}).
Though we omitted the task graph in the diagram (\ref{eq:Seq_GS_GO}), we believe we are considering a task graph, taking into account the hierarchical structure at each level.
\begin{equation}
\xymatrix{ 
\cdots  & G_{S,i}  \ar[r]^{p_i}\ar[ld]^{p^{(i)}}\ar[dd]|{q_{i}} 
 &G_{S,i+1} \ar[rd]^{p^{(i+1)}} \ar[dd]|{q_{i+1}}&\cdots \\
G_{T,i}\ar[rd]^{i^{(i)}} &    &  &G_{T,i+1}
\ar[ld]^{i^{(i+1)}}\\
\cdots & G_{O,i}  \ar[r]^{p'_i}& G'_{O,i+1}  &\cdots 
 }
\label{eq:Seq_GS_GO2}
\end{equation}
This allows us to consider the geometric structure for designing software from each individual to the entire organization, consistent with the hierarchical structure involved in administration.

\end{enumerate}

\begin{remark}\label{rmk:Conwayfinal}
{\rm{
These are correspondences that could not be seen by simply looking at the graph homomorphism between $G_S$ and $G_O$. 
They are the ones that came into view with the continuous map of the topology of the graphs.
Even if a continuous map or w-homomorphism exists, it is not uniquely determined in general, so operations like what kind of continuous map or w-homomorphism to choose seem to be {\lq}management{\rq} itself.

What is the novelty of our continuous map instead of homomorphism? 
The answer is that it is a bit more flexible than {\lq}one module and one team (individual){\rq}, allowing {\lq}one-to-many{\rq}, 
{\lq}multiple teams corresponding to a particular relation (edge) of modules{\rq}, and {\lq}hierarchical structure{\rq}.

Apart from the mathematical details, we believe that this approach is useful or available from a software development management perspective.
}}
\end{remark}

\section{Conclusion}

In this article, we revise Conway's Law from a mathematical point of view.
By introducing a task graph, we first rigorously state Conway's law based on the homomorphism in graph theory in Law \ref{law1}.
The task graph is a graph that represents the geometric structure of tasks.
For cognitive load and avoidance of unnecessary communication, the task graph is a nice tool to be considered.

Though Conway mentioned the importance of the hierarchical treatments of these geometrical structures, it cannot be expressed by the homomorphism well.
Thus second, we introduce w-homomorphism for in Law \ref{law2} and then we can describe the internal communication in the team itself, and partially express the hierarchy of the systems and the organizations.

For more natural expressions of the hierarchical treatments of the systems and the organizations with the Conway's correspondences, we introduce the graph topology and continuous map.
Such expressions is recent required.
We have rephrased these statements in terms of the continuous maps in graph topology in Law \ref{law3}.

To use graph topology and the continuous map in Conway's law, we have prepared them as mathematical tools.
We can express the case where a module or a subsystem can be made responsible by a team or an association of teams that has internal geometric properties.
Due to secure, knowledge hiding and so on, the communication of certain teams should be made responsible by a special team corresponding to the feature of the connections of submodules, it can be expressed in terms of the continuous map.

For future applications, we envision the following research directions.
\begin{enumerate}

\item This mathematical foundation only assumes the continuity between the two graphics. We can generalize the software to be a {\lq\lq}collaborative artifact{\rq\rq} that includes text, music, and artwork, or even a more general product.

\item
It would be useful to define what the {\lq\lq}edges{\rq\rq} of the graph in the software and the organization mean in practice. Terms such as dependencies, cohesion, coupling, utility in the software graphs, and communication, knowledge, information flow, cognitive load, responsibility, and collaboration in the organization graphs can characterize the edge definition to apply to the software engineering field. 
We believe that our scheme will make the patterns described in the context of Team Topologies \cite{SP2019} more understandable and explainable with rigorous mathematical foundations.
(The terminology of {\lq}team topologies{\rq} \cite{SP2019} (not {\lq}team graphs{\rq}) can be justified from a general topology.)

\item 
For given the complicated geometric structure of the organization connected with administration and so on, and the variation of the geometric structure of software for time development, we should design the organization and these cards with hierarchical structures. 
Our mathematical tools express them.

\item 
Since we have topological spaces associated with graphs, it is natural to consider the sheaves over them \cite{Gogunen, Sendroiu}.
We can argue more precise structures on them.

\item 
Though we consider a relation between a single system and the organization, our scheme is extended to the relation between multiple systems and the organization, i.e., the $G_S$ is not a connected graph.
In general, the organization has several software developments at the same time, and the task graph also works well as a nice tool to express this.

\end{enumerate}

\bigskip
\subsection*{Acknowledgments:}
We are grateful to Professor Kenichi Tamano for comments.
We are also grateful to Dr.~Elena \c{S}endroiu for pointing out several typographical errors and for comments on the earlier version.
This work was supported by Institute of Mathematics for Industry, Joint Usage/Research Center in Kyushu University. 
``IMI workshop I: Geometry and Algebra in Material Science III'',
September 8--10, 2022" (2022a003)).
S. M. has been supported by the Grant-in-Aid for Scientific Research (C) of Japan Society for the Promotion of Science Grant, No.21K03289, and 
S. O. has been supported by Grants: Young Scientist of Japan Society for the Promotion of Science Grant, no. 22K13963.

\bigskip


\bigskip
\bigskip

\noindent
Shigeki Matsutani\\
Electrical Engineering and Computer Science,\\
Graduate School of Natural Science \& Technology, \\
Kanazawa~University,\\
Kakuma Kanazawa, 920-1192, Japan\\
\texttt{ORCID:0000-0002-2981-0096}\\
\texttt{s-matsutani@se.kanazawa-u.ac.jp}

\bigskip
\noindent
Shousuke Ohmori,\\
National Institute of Technology, Gumma College\\
Electrical Engineering and Computer Science,\\
Maebash, 371-8530, Japan\\
\texttt{ORCID:0000-0002-1623-9994}\\
\texttt{42261timemachine@ruri.waseda.jp}

\bigskip
\noindent
Kenji Hiranabe,\\
ESM, Inc.\\
3-111 Toiya-cho, Fukui, 918-8231, Japan\\ 
\texttt{k-hiranabe@esm.co.jp}

\bigskip
\noindent
Eiichi Hanyuda\\
Mamezou Co. , Ltd.\\
Shinjuku-Mitsui building 34F, 2-1-1, \\
Nishishinjuku Shinjuku-ku, Tokyo, 163-0434 Japan,\\
Information-technology Promotion Agency Japan \\
2-28-8 Honkomagome, Bunkyo-ku, Tokyo, 113-6591, Japan \\
\texttt{hanyuda@mamezou.com}


\begin{thebibliography}{AAAA}


\bibitem[AN]{AN2023}
\by{A. Aniyany and S. Naduvath}
\paper{A study on graph topology}
\jour{Comm. Combi. Opt.} \vol{8}
\yr{2023} \pages{397--409}

\bibitem[BR]{BR}
\by{R. Balakrishnan ? K. Ranganathan}
\book{A textbook of graph theory, 2nd ed.}
Springer, Berlin, 2012.


\bibitem[C]{Conway}
\by{M. E. Conway}
\paper{How Do Committees Invent?}
\jour{Datamation magazine,} 
Thompson Publications, Inc. \yr{1968}.


\bibitem[D]{Diestel}
\by{R. Diestel}
\book{Graph theory, fifth ed.}
Springer, Berlin. 2017

\bibitem[G]{Gogunen}
\by{J. A. Goguen}, 
\paper{Sheaf semantics for concurrent interacting objects}
\jour{Math. Stru, Comp. Sci.} \vol{2} \yr{1991} \pages{159-191}.

\bibitem[HG]{HG1999}
\by{J. D. Herbsleb and R. E. Grinter}
\paper{Splitting the Organization and Integrating the Code:
Conway's Law Revisited}
\jour{ICSE '99: Proceedings of the 21st international conference on Software engineering} \yr{1999}  \pages{85--95}.

\bibitem[KCD]{KCD2012}
\by{I. Kwan, M. Cataldo, and D. Damian}
\paper{Conway's Law Revisited: The Evidence for a Task-Based Perspective}
\jour{IEEE Software} \vol{29} \yr{2012} 90-93

\bibitem[K]{Kelley}
\by{J. L. Kelley}
\paper{General Topology}
(Springer New York, NY, 1955)

\bibitem[OYYK]{Ohmori}
\by{S. Ohmori, Y. Yamazaki, T. Yamamoto, and A. Kitada}
\paper{Universal topological representation of geometric patterns}
\jour{Physica Scripta} \vol{94} Number 10 Citation Shousuke Ohmori et al 2019 Phys. Scr. 94 105213

\bibitem[Sc]{Schofield}
N Schofield
\book{Mathematical Methods in Economics and Social Choice}
\publ{Springer}
2014.

\bibitem[Se]{Sendroiu}
\by{E. \c{S}endroiu}
\by{Sheaf Tools for Computation}
\jour{J. Appl. Math. Comp.}
\vol{184} \yr{2007} \pages{131-141}.

\bibitem[SP]{SP2019}
\by{M. Skelton, P. M. Pais}
\book{Team Topologies: Organizing Business and Technology Teams for Fast Flow}
\publ{IT Revolution Press, LLC} 2019

\bibitem[St]{St2023}
\by{D. St{\aa}hl}
\paper{The dynamic versus the stable team: The unspoken question in
large-scale agile development}
\jour{J Softw Evol Proc.} \yr{2023}
 \vol{2023}  \pages{2023;e2589, 1-23}.

\bibitem[W]{Weil}
\by{A. Weil} 
\paper{Appendix to part one?: on the algebraic study of certain types of marriage laws (Murngin system)}
 Levi-Strauss, C. 
\book{The Elementary Structures of Kinship (Les Structures Elementaires de La Parente); Revised Edition} Translated from the French by J.H. Bell and J.R. von Sturmer; R. Needham, Editor -- London; 
Eyre and Spottiswood , 1969.


\end{thebibliography}
\end{document}